\documentclass[12pt]{article}

\usepackage[letterpaper,top=1in,bottom=1in,left=1in,right=1in]{geometry}
\usepackage{upgreek}

\usepackage{titlesec}
\titleformat{\section}{\normalfont\scshape\large}{\S\thesection.}{0.5em}{}
\titleformat{\subsection}{\normalfont\scshape\normalsize}{\thesubsection.}{0.5em}{}
\titleformat{\subsubsection}{\normalfont\itshape\normalsize}{\thesubsubsection.}{0.5em}{}

\usepackage{amsthm}
\newtheorem{theorem}{Theorem}[section]
\newtheorem{lemma}[theorem]{Lemma}

\theoremstyle{definition}
\newtheorem{definition}[theorem]{Definition}
\newtheorem{example}[theorem]{Example}

\theoremstyle{remark}
\newtheorem{remark}[theorem]{Remark}

\usepackage[numbers]{natbib}
\title{Loop-Checking and Counter-Model Extraction for Intuitionistic Tense Logics via Nested Sequents}
\author{Tim S. Lyon\\
\small Technische Universit{\"a}t Dresden\\
\small \texttt{timothy\_stephen.lyon@tu-dresden.de}}
\date{}

\usepackage{amssymb, amsmath}
\usepackage{mathtools}
\usepackage{bussproofs, rotating, xcolor, upgreek}
\usepackage{comment}
\usepackage[ruled, linesnumbered]{algorithm2e}
\usepackage{mathrsfs}

\usepackage{tikz, tikz-cd, color}
\usetikzlibrary{arrows, automata}
\usetikzlibrary{trees}
\usetikzlibrary{fit}
\usetikzlibrary{positioning,arrows,calc,decorations.markings}
\tikzset{
modal/.style={>=stealth',shorten >=1pt,shorten <=1pt,auto,node distance=1.5cm,semithick},
world/.style={circle,draw,minimum size=0.5cm,fill=gray!15},
point/.style={circle,draw,inner sep=0.5mm,fill=black},
reflexive above/.style={->,loop,looseness=7,in=120,out=60},
reflexive below/.style={->,loop,looseness=7,in=240,out=300},
reflexive left/.style={->,loop,looseness=7,in=150,out=210},
reflexive right/.style={->,loop,looseness=7,in=30,out=330}%,
}

\usepackage{enumitem}

\usepackage{changepage,enumerate,relsize}
\usepackage{multicol}

\usepackage{wrapfig}
\theoremstyle{definition}
\theoremstyle{remark}
  \newtheorem{innercustomthm}{Theorem}
\newenvironment{customthm}[1]
  {\renewcommand\theinnercustomthm{#1}\innercustomthm}
  {\endinnercustomthm}

\DeclareSymbolFont{extraup}{U}{zavm}{m}{n}
\DeclareMathSymbol{\vardiamond}{\mathalpha}{extraup}{87}

\definecolor{tim}{RGB}{0, 0, 250}

\newcommand{\iffi}{\textit{iff} }
\newcommand{\etc}{$\ldots$ }
\newcommand{\dfn}{Definition}
\newcommand{\fig}{Figure}
\newcommand{\lem}{Lemma}

\newcommand{\D}{\mathsf{D}}
\newcommand{\T}{\mathsf{T}}

\newcommand{\B}{\mathsf{B}}

\newcommand{\axs}{\mathcal{A}}

\newcommand{\sbl}{\{}
\newcommand{\sbr}{\}}

\newcommand{\nikt}{\mathsf{NIK_{t}}}
\newcommand{\niktc}{\mathsf{NIK_{t}C}}

\newcommand{\calc}{\mathsf{NIK_{t}C}}

\newcommand{\ns}{\mathcal{G}}
\newcommand{\nsii}{\mathcal{H}}
\newcommand{\nsiii}{\mathcal{K}}
\newcommand{\nsiv}{\mathcal{J}}
\newcommand{\nsinp}{\mathcal{I}}

\newcommand{\sar}{\vdash}

\newcommand{\sat}{\vDash} %\Vdash^{\albet}}

\newcommand{\da}{\downarrow}

\newcommand{\prf}{\pi} %\mathcal{D}}
            
\newcommand{\ikt}{\mathsf{IK_{t}}}
\newcommand{\iktc}{\mathsf{IK_{t}C}}

\newcommand{\isf}{\mathsf{IS4}}

\newcommand{\fd}{\circ} %\mathsf{F}}
\newcommand{\bd}{\bullet} %\mathsf{P}}
\newcommand{\forb}{\mathsf{x}} %\circledast}

\newcommand{\charx}{x}

\newcommand{\conv}[1]{\overline{#1}}

\newcommand{\iimp}{\rightarrow} %\supset}
\newcommand{\imp}{\rightarrow}
\newcommand{\inot}{\neg}

\newcommand{\xbox}{[\forb]}

\newcommand{\xdia}{\langle \forb \rangle}

\newcommand{\wbox}{{\scriptstyle\square}}
\newcommand{\wdia}{{\mathbin{\text{\rotatebox[origin=c]{45}{$\scriptstyle\square$}}}}} %\diamondsuit}

\newcommand{\bbox}{{\scriptstyle\blacksquare}}
\newcommand{\bdia}{\mathbin{\text{\rotatebox[origin=c]{45}{$\scriptstyle\blacksquare$}}}}

\newcommand{\lang}{\mathscr{L}}
\newcommand{\langt}{\mathscr{L}_{T}}

\newcommand{\atm}{\mathsf{Atm}}

\newcommand{\sufo}{\mathsf{sufo}}

\newcommand{\id}{\mathsf{id}}

\newcommand{\boximpr}{[\imp]\mathsf{R}}

\newcommand{\conl}{\land\mathsf{L}}
\newcommand{\conr}{\land\mathsf{R}}
\newcommand{\disl}{\lor\mathsf{L}}
\newcommand{\disr}{\lor\mathsf{R}}

\newcommand{\iimpl}{{\iimp}\mathsf{L}}

\newcommand{\iimpr}{{\iimp}\mathsf{R}}
\newcommand{\wdial}{\wdia\mathsf{L}}

\newcommand{\bdial}{\bdia\mathsf{L}}
\newcommand{\wboxr}{\wbox\mathsf{R}}
\newcommand{\bboxr}{\bbox\mathsf{R}}
\newcommand{\botl}{\bot\mathsf{L}}

\newcommand{\ddr}{\mathsf{d}}

\newcommand{\names}{\mathtt{Names}}

\newcommand{\nshole}{\{\cdot\}}

\newcommand{\mint}{\iota}
\newcommand{\frames}{\mathscr{F}}

\newcommand{\simplepath}[1]{\overset{\mathclap{#1}}{\leadsto}}
\newcommand{\prpath}[1]{\hspace{3pt}\raise-3pt\hbox{$\simplepath{#1}$}\hspace{3pt}}

\newcommand{\set}[1]{\{#1\}}

\newcommand{\lb}{(} %\langle}
\newcommand{\rb}{)} %\rangle}
\newcommand{\ru}{\mathsf{r}}

\newcommand{\proves}{\Vdash}

\newcommand{\branch}{\mathcal{B}}

\newcommand{\infs}[1]{\mathsf{In}(#1)} %InF
\newcommand{\outfs}[1]{\mathsf{Out}(#1)} %OutF
\newcommand{\cons}{\mathsf{C}}

\newcommand{\depk}[2]{\mathsf{dp}_{#2}(#1)}
\newcommand{\mdep}[1]{\mathsf{md}(#1)}

\newcommand{\idc}{(id)}
\newcommand{\botlc}{(\bot L)}
\newcommand{\dislc}{(\lor L)}
\newcommand{\disrc}{(\lor R)}
\newcommand{\conlc}{(\land L)}
\newcommand{\conrc}{(\land R)}

\newcommand{\iimplc}{({\iimp} L)}
\newcommand{\ddrc}{(d)}
\newcommand{\xdialc}{(\xdia L)}
\newcommand{\xboxlc}{(\xbox L)}
\newcommand{\xdiarc}{(\xdia R)}

\newcommand{\bp}{\mathtt{bp}} 
\newcommand{\bpv}{\ctv_{\bp}}
\newcommand{\bpe}{\cte_{\bp}}

\newcommand{\bps}{\mathcal{S}}

\newcommand{\ct}{\mathtt{ct}}
\newcommand{\ctv}{\mathcal{V}}
\newcommand{\cte}{\prec}
\newcommand{\ctl}{\mathcal{L}}

\newcommand{\hequiv}{\simeq}
\newcommand{\hec}[1]{[\![#1]\!]}
\newcommand{\enum}{\eta}

\newcommand{\mor}{\gamma}
\newcommand{\sthom}{\alpha} %h_{s}}
\newcommand{\wkhom}{\beta} %h_{w}}
\newcommand{\vechom}{\vec{\sthom}}
\newcommand{\homnat}{\bar{\wkhom}}

\newcommand{\hreach}[1]{\leadsto^{#1}} %\overset{#1}{\longrightarrow}} %\tikztwoheadrightarrow}} %\xtwoheadrightarrow{#1}{}}

\newcommand{\tr}[1]{T_{#1}}
\newcommand{\ltr}{T}
\newcommand{\trv}{V}
\newcommand{\tre}{\lessdot} %<_{\ltr}}
\newcommand{\trep}[1]{\lessdot_{#1}}
\newcommand{\trl}{L}

\newcommand{\X}{\mathsf{X}}
\newcommand{\rable}{\twoheadrightarrow}

\newcommand{\rableosf}{\twoheadrightarrow_{\fd}}

\newcommand{\rableosx}{\twoheadrightarrow_{\forb}^{\cons}}

\newcommand{\prove}{\mathtt{Prove}_{\cons}}
\newcommand{\true}{\mathtt{True}}
\newcommand{\false}{\mathtt{False}}

\newcommand{\dbr}{\mathsf{db}}
\newcommand{\tval}{\mathtt{T}}
\newcommand{\fval}{\mathtt{F}}

\newcommand{\pathi}{w} %\boldsymbol{w}}
\newcommand{\pathii}{u} %\boldsymbol{u}}
\newcommand{\pathiii}{v} %\boldsymbol{v}}
\newcommand{\pathiv}{v'} %\boldsymbol{v}'}

\newcommand{\prgr}[1]{\mathsf{pg}(#1)}
\newcommand{\rt}{w}

\newcommand{\bptr}[1]{\enum\tr{#1}}

\newcommand{\fcons}{\frames_{\cons}}

\newcommand{\xdial}{\langle \forb \rangle\mathsf{L}}
\newcommand{\xdiar}{\langle \forb \rangle\mathsf{R}}
\newcommand{\wdiar}{\wdia\mathsf{R}}
\newcommand{\bdiar}{\bdia\mathsf{R}}
\newcommand{\xdiari}{\langle \forb \rangle\mathsf{R}_{1}}
\newcommand{\xdiarii}{\langle \forb \rangle\mathsf{R}_{2}}
\newcommand{\xboxl}{[\forb]\mathsf{L}}
\newcommand{\bboxl}{\bbox\mathsf{L}}
\newcommand{\wboxl}{\wbox\mathsf{L}}
\newcommand{\xboxli}{[\forb]\mathsf{L}_{1}}
\newcommand{\xboxlii}{[\forb]\mathsf{L}_{2}}
\newcommand{\xboxr}{[\forb]\mathsf{R}}

\begin{document}

\maketitle
\thispagestyle{empty}

% ABSTRACT
\noindent
\begin{minipage}{\textwidth}
\noindent
\textbf{Abstract.} This paper develops a novel nested sequent proof-search methodology for intuitionistic tense logics (ITLs), supporting finite counter-model extraction. We introduce a new loop-checking method that detects repeating nested sequents using homomorphisms, thereby bounding the height of derivations during proof-search. Due to the non-invertibility of some inference rules, the algorithm does not construct a single derivation, but a generalized structure we call a \emph{computation tree}. We show how proofs and counter-models can be extracted from computation trees when proof-search succeeds or fails, respectively. This establishes the finite model property for each ITL of the form $\ikt \cup \axs$, with $\axs \subseteq \set{\T,\B,\D}$.
\end{minipage}

%%%INTRO
\section{Introduction}\label{sec:intro}

%Notes:
%1. To avoid loop checking we make use of finite tree calculus with diagonal formula in box rule -> automatic termination of modal relation
%2. Cite papers on int gram logics and related calculi

Numerous authors have proposed logics that provide an intuitionistic treatment of modal reasoning~\cite{BovDov84,Bul65,Dov85,Fit48,PloSti86,Fis77,Fis84,Sim94}; see Stewart et al.~\cite{StePaiAle15} for an overview. Among these, the formulations due to Fischer Servi~\cite{Fis84}, Plotkin and Stirling~\cite{PloSti86}, and Ewald~\cite{Ewa86} have been particularly influential, especially following Simpson’s work~\cite{Sim94}, which placed such logics on a solid philosophical footing. Simpson’s key insight was to interpret modal operators within an intuitionistic meta-theory. This perspective carries core intuitionistic principles into the modal setting--such as the disjunction property and the failure of the law of excluded middle--and disrupts the classical duality between necessity and possibility. Beyond their foundational significance, intuitionistic modal logics (IMLs) have found applications in computer science, including program verification~\cite{FaiMen95}, reasoning about functional programs~\cite{Pit91}, and the design of programming languages~\cite{DavPfe01}.

Ewald’s intuitionistic tense logics (ITLs) form a prominent sub-class of IMLs that integrate constructive and temporal reasoning within a single framework. This is achieved by extending intuitionistic logic with temporal operators quantifying over both past and future states in Kripke-style models~\cite{Ewa86}. Recently, proof systems for a broad class of ITLs were introduced in the formalism of nested sequents~\cite{Lyo25}, though proof-search methods for these systems remain completely unexplored.

Designing proof-search procedures that support counter-model extraction for IMLs is notoriously difficult for two main reasons. First, these logics are semantically characterized by bi-relational Kripke models in which the intuitionistic accessibility relation is transitive. This transitivity can lead to the non-termination of proof search in sequent calculi, necessitating sophisticated loop-checking mechanisms to ensure termination. Second, sequent systems for intuitionistic modal logics make pervasive use of non-invertible inference rules (e.g.,~\cite{KuzStr19,Sim94,Str13,Lyo21a}), rendering the standard paradigm of constructing a \emph{single} derivation witnessing provability or refutability inapplicable. In the intuitionistic setting, non-invertibility introduces intrinsic branching in proof search, obscuring both the extraction of proofs from successful searches and the construction of counter-models from failed ones. In this paper, we overcome both obstacles and present proof-search algorithms for ITLs that support finite counter-model extraction.

While Simpson proved decidability for several extensions of $\mathsf{IK}$ via proof-search in labeled sequent calculi, his method does not yield counter-models when proof search fails. Instead, his procedure shows that only finitely many candidate derivations exist for a given formula and then checks whether any of them is a valid proof. The basic logic $\mathsf{IK}$ has more recently been shown to admit proof-search in a labeled sequent calculus that \emph{does} support counter-model extraction~\cite{GirEtAl24}, but little is known beyond this case.\footnote{The paper~\cite{GirEtAl23} attempted to establish decidability of $\isf$ via labeled proof-search, but the proof is known to contain an error (personal communication).}

To bridge this gap, we develop a novel proof-search methodology for Ewald's $\ikt$ extended with any set of axioms $\axs \subseteq \set{\T, \B, \D}$, supporting finite counter-model extraction. The ability to extract finite counter-models from failed proof search establishes the finite model property (FMP), a key meta-theoretic property with consequences for expressivity, complexity, and automated reasoning. We illustrate the inner workings of this methodology within the nested sequent formalism.

%\footnote{Leivant~\cite[p.~361]{Lei81} had already introduced a notational variant of nested sequents in his proof-theoretic work on propositional dynamic logic, where formulae are prefixed by so-called execution sequences.} 

The nested sequent formalism was introduced independently by Kashima~\cite{Kas94} and Bull~\cite{Bul92}, though later work by Br\"unnler~\cite{Bru09} and Poggiolesi~\cite{Pog09} was instrumental to the popularity of the formalism. A nested sequent is a tree of Gentzen sequents, and thus generalizes Gentzen's original formalism~\cite{Gen35a,Gen35b}. Such sequents yield elegant calculi that minimize syntactic overhead, produce compact proofs, and where termination of proof-search is more easily established (cf.~\cite{Lyo21thesis,LyoOst23}). Moreover, such systems have proven well-suited for extracting counter-models from failed proof-search for classical modal logics~\cite{Bru09,LyoGom22,TiuIanGor12}. These favorable computational properties arise in part from the \emph{analyticity} of the systems--i.e., their satisfaction of the \emph{subformula property}--which ensures that every formula %appearing 
in the premise of a rule is a subformula of its conclusion.

The nested systems used in this paper are variants of systems that already exist in the literature. In~\cite{Lyo22,Lyo25}, \emph{single-conclusioned} nested systems were provided for intuitionistic grammar logics~\cite{Lyo21b}, which is a class of logics subsuming %$\ikt \cup \axs$ for $\axs \subseteq \set{\T,\B,\D}$. 
the ITLs considered here. By contrast, the systems in this paper are \emph{multi-conclusioned}, meaning, nested sequents may contain multiple formulae in the consequent. Kuznets and Stra{\ss}burger~\cite{KuzStr19} introduced multi-conclusioned systems for standard IMLs without temporal modalities, so our systems can be viewed as generalizations of these. %We note that, beyond nested sequent systems, various \emph{labeled sequent systems}--which use graphs of Gentzen sequents in proofs--have been provided for IMLs and ITLs~\cite{Lyo22,MarMorStr21,Sim94}.

\paragraph{Contributions.} This paper makes the following contributions:

$\bullet$ We present a correct and terminating proof-search algorithm for \emph{non-invertible} nested sequent systems. In such systems, not all rules are invertible; that is, the validity of a rule's conclusion does not necessarily imply the validity of its premises. Invertibility typically enables counter-model extraction from a single branch of a derivation obtained by failed proof-search, thereby simplifying the construction of counter-models. Lacking this property, we must adopt an alternative strategy. Our solution is to forego the construction of a single derivation and instead build a more general structure, which we call a \emph{computation tree}. This tree compactly represents all possible derivations of the input. %To the best of our knowledge, this 
This constitutes the first proof-search algorithm developed specifically for non-invertible nested sequent systems.

%$\bullet$ Our proof-search algorithm is \emph{uniform} over the logics we consider, that is, one does not need to consider ad hoc cases to establish correctness or termination of proof-search based on idiosyncrasies of the logics. %We also conjecture that this method can be easily generalized to uniformly establish decidability for $\ik \cup \axsii$ and $\ikt \cup \axsii$ for each $\axsii \subseteq \axs$

$\bullet$ An important aspect of our algorithm is that it performs loop-checking with nested sequents. We introduce a new technique that relies on certain homomorphisms to detect repeating nested sequents in branches of a computation tree.

$\bullet$ We show how computation trees can be pruned to recover proofs when proof-search succeeds and deems the input valid. Moreover, we show how finite bi-relational models can be extracted from computation trees when proof-search fails and the input is deemed invalid. While nested sequents can be used to define the modal relation in a bi-relational model, the intuitionistic relation is defined by sequences of homomorphisms between nested sequents. %computation tree.

$\bullet$ Finally, in his 1986 paper, Ewald~\cite{Ewa86} attempted to prove the finite model property of $\ikt \cup \axs$, for $\axs \subseteq \set{\T,\B,\D}$, by means of a filtration argument. As noted by Simpson~\cite[p.~158]{Sim94}, an error in Ewald’s proof was identified by Colin Stirling. To the best of our knowledge, this proof has not been subsequently corrected in the literature. In this paper, we provide a correct proof of the finite model property for these intuitionistic tense logics. % \tim{which gives a $\twoexptime$ upper bound on the complexity of each logic.}

% $\ikt \cup \axs$ with $\axs \subseteq \set{\T,\B,\D}$.

%Following is not true actually:
%The counter-models exhibit a `quasi-finite model property,' namely, the intuitionistic relation is finite, while the modal relation can be infinite.

\paragraph{Paper Organization.} In Section~\ref{sec:prelims}, we present the preliminaries for ITLs. % $\ikt \cup \axs$ with $\axs \subseteq \set{\T,\B,\D}$. 
Section~\ref{sec:systems} introduces nested sequent systems and discusses soundness and completeness. Section~\ref{sec:decid} presents our proof-search algorithm and establishes its correctness and termination. % \tim{in $\twoexptime$}. 
Finally, in Section~\ref{sec:conclusion}, we discuss directions for future work.

%%%PRELIMINARIES
\section{Preliminaries}\label{sec:prelims}

%\resizebox{\columnwidth}{!}{
%%%NOTES

%\subsection{Syntax and Semantics}

We let $\atm := \{p, q, r, \ldots\}$ be a denumerable set of \emph{atoms} and define the tense language $\lang$ via the following grammar in BNF:
$$
A ::= p \mid \bot \mid A \lor A \mid A \land A \mid A \iimp A \mid \wdia A \mid \bdia A \mid \wbox A \mid \bbox A
$$
where $p$ ranges over the set $\atm$ of atoms. We refer to $\wdia$ and $\wbox$ as \emph{forward} modalities, and $\bdia$ and $\bbox$ as \emph{backward} modalities. We use $A$, $B$, $C$, \etc to denote formulae in $\lang$ and define %$\top := \bot \iimp \bot$, 
$\inot A := A \iimp \bot$ and $A \leftrightarrow B := (A \iimp B) \land (B \iimp A)$, as usual. The \emph{length} of a formula $A$, denoted $\ell(A)$, is the number of symbols contained in $A$. We define the set of \emph{subformulae} of $A$, denoted $\sufo(A)$, as usual; $B$ is a \emph{subformula} of $A$ \iffi $B \in \sufo(A)$. We define the \emph{modal depth} of a formula $A$, denoted $\mdep{A}$, recursively as follows:
\begin{description}

\item[$\bullet$] $\mdep{p} = \mdep{\bot} := 0$ with $p \in \atm$;

\item[$\bullet$] $\mdep{A \circ B} := \max\set{\mdep{A},\mdep{B}}$ with $\circ \in \set{\lor, \land, \iimp}$;

\item[$\bullet$] $\mdep{\nabla A} := \mdep{A} + 1$ with $\nabla \in \set{\wdia,\bdia,\wbox,\bbox}$.

\end{description}
We define $\mdep{\Gamma} := \max\set{\mdep{A} \mid A \in \Gamma}$ for a finite set $\Gamma$ of formulae. Formulae are interpreted over \emph{bi-relational frames} and \emph{models} (cf.~\cite{Ewa86}), which we refer to as \emph{frames} and \emph{models} for simplicity.

\begin{definition}[Frame]\label{def:bi-relational-model} A \emph{frame} is a tuple $F = (W, \leq, R)$ where:
\begin{description}%[leftmargin=!, labelwidth=0em, labelsep=.25em]

\item[$\bullet$] $W$ is a non-empty set of \emph{worlds} $\{w, u, v, \ldots\}$;

\item[$\bullet$] The \emph{intuitionistic relation} $\leq \ \subseteq W \times W$ is a pre-order; %, i.e., it is reflexive and transitive;

\item[$\bullet$] The \emph{modal relation} $R \subseteq W \times W$ satisfies:

\end{description}
\begin{description}[leftmargin=!, labelwidth=0em, labelsep=.25em]

\item[$(F1)$] For all $w, w', v \in W$, if $w \leq w'$ and $w R v$, then there exists a $v' \in W$ such that $w' R v'$ and $v \leq v'$;

\item[$(F2)$] For all $w, v, v' \in W$, if $w R v$ and $v \leq v'$, then there exists a $w' \in W$ such that $w \leq w'$ and $w' R v'$.

%\item[(F3)] $w R_{\charx} u$ \iffi $u R_{\conv{\charx}} w$.

\end{description}
\end{definition}

\begin{definition}[Model]\label{def:bi-relational-model} We define a \emph{model} based on a frame $F$ to be a pair $M = (F, V)$ such that 
%\begin{itemize}
%\item $F = (W, \leq, \{R_{\charx} \mid \charx \in \albet\})$ is a frame;
$V : W \to 2^{\atm}$ is a \emph{valuation function} satisfying the \emph{monotonicity condition}: (M) for each $w, u \in W$, if $w \leq u$, then $V(w) \subseteq V(u)$.
%\end{itemize}
%We will often times write a model $M$ as $(W, \leq, \{R_{\charx} \mid \charx \in \albet\},V)$.
\end{definition}

The (F1) and (F2) conditions %are depicted in \fig~\ref{fig:f1-f2} and 
ensure the monotonicity of complex formulae (see \lem~\ref{lem:persistence}) in our models, which is a property characteristic of intuitionistic logics. See Simpson~\cite[Section~3]{Sim94} for a discussion of these conditions.

%We interpret formulae from $\lang$ over models via the following clauses.

%\cite{GabSheSkv09}
\begin{definition}[Semantic Clauses]\label{def:semantic-clauses} Let $M = (W, \leq, R, V)$ be a model with $w \in W$. The \emph{satisfaction relation} $M,w \sat A$ between $w \in W$ and a formula $A \in \lang$ is inductively defined as follows:
\begin{description}[leftmargin=!, labelwidth=0em, labelsep=.25em]

\item[$\bullet$] $M,w \sat p$ \iffi $p \in V(w)$, for $p \in \atm$;

\item[$\bullet$] $M,w \not\sat \bot$;

\item[$\bullet$] $M,w \sat A \lor B$ \iffi $M,w \sat A$ or $M,w \sat B$;

\item[$\bullet$] $M,w \sat A \land B$ \iffi $M,w \sat A$ and $M,w \sat B$;

%\item $M, w \sat \inot A$ \iffi for all $w' \in W$, if $w \leq w'$, then $M,w' \not\sat A$;

\item[$\bullet$] $M,w \sat A \iimp B$ \iffi $\forall u \in W$, if $w \leq u$ and $M,u \sat A$, then $M,u \sat B$;

\item[$\bullet$] $M,w \sat \wdia A$ \iffi $\exists u \in W$ such that $wRu$ and $M,u \sat A$;

\item[$\bullet$] $M,w \sat \bdia A$ \iffi $\exists u \in W$ such that $uRw$ and $M,u \sat A$;

\item[$\bullet$] $M,w \sat \wbox A$ \iffi $\forall u, v \in W$, if $w \leq u$ and $u R v$, then $M,v \sat A$;

\item[$\bullet$] $M,w \sat \bbox A$ \iffi $\forall u, v \in W$, if $w \leq u$ and $v R u$, then $M,v \sat A$;

\item[$\bullet$] $M \sat A$ \iffi $\forall w \in W$, $M, w \sat A$.

\end{description}
%\end{itemize}
%A formula $A$ is defined to be \emph{globally true on $M$}, written $M \sat A$, \iffi $M,u \sat A$ for all worlds $u \in W$ of $M$. A formula $A$ is defined to be \emph{valid}, written $\sat A$, \iffi $A$ is globally true on every bi-relational $\albet$-model. Last, we say that a set $\fseti$ of formulae \emph{semantically implies} a formula $A$, written $\fseti \sat A$, \iffi for all bi-relational $\albet$-models $M$ and each $w \in W$ of $M$, if $M, w \sat B$ for each $B \in \fseti$, then $M,w \sat A$.
A formula $A$ is \emph{valid} relative to a class of frames $\frames$, written $\frames \sat A$, \iffi for all models $M$ based on a frame $F \in \frames$, $M \sat A$.
\end{definition}

%The following %can be proven by a straightforward induction on the length of $A$ 
%is a well-known property of ITLs~\cite{Ewa86}.

\begin{lemma}[\cite{Ewa86}]\label{lem:persistence}
Let $M$ be a model with $w,u \in W$. If $w \leq u$ and $M, w \sat A$, then $M, u \sat A$.
\end{lemma}

We consider three conditions that may be imposed on frames, namely, reflexivity ($\T$), symmetry ($\B$), and seriality ($\D$), listed below:
\begin{description}%[leftmargin=!, labelwidth=0em, labelsep=.25em]

\item[$(\T)$] for all $w \in W$, $wRw$;

\item[$(\B)$] for all $w,u \in W$, if $wRu$, then $uRw$;

\item[$(\D)$] for all $w \in W$, there exists a $u \in W$ such that $wRu$.

\end{description}
We let $\cons \subseteq \set{\T,\B,\D}$ and let $\fcons$ be the set of all frames satisfying the conditions in $\cons$. We define the intuitionistic tense logic $\iktc := \set{A \in \lang \mid \fcons \sat A}$. Since reflexivity implies seriality, this gives rise to six distinct ITLs: $\ikt$, $\ikt\mathsf{D}$, $\ikt\mathsf{B}$, $\ikt\mathsf{T}$, $\ikt\mathsf{DB}$, and $\ikt\mathsf{TB}$. Each logic is known to admit a sound and complete axiom system~\cite{Ewa86}. %We do not recall these axiom systems here (cf.~\cite{Ewa86,Fis84,PloSti86})

%with the $\T$ axiom $(A \iimp \wdia A) \land (\wbox A \iimp A)$ and $\four$ axiom $(\wdia \wdia A \iimp \wdia A) \land (\wbox A \iimp \wbox \wbox A)$.

%%%CALCULUS
\section{Nested Sequent Systems}\label{sec:systems}

%NOTES:

We let $\Gamma$, $\Delta$, $\Sigma$, $\ldots$ be finite sets of formulae from $\lang$ and define nested sequents inductively as follows:

\smallskip

\noindent
(1) Each \emph{Gentzen sequent} of the form $\Gamma \sar \Delta$ is a nested sequent;

\noindent
(2) Each expression $\Gamma \sar \Delta, (\forb_{1})[\ns_{1}], \ldots, (\forb_{n})[\ns_{n}]$ is a nested sequent if $\ns_{i}$ is a nested sequent and $\forb_{i} \in \set{\fd,\bd}$.\footnote{We let $i \in [n]$ be a shorthand for $1 \leq i \leq n$.}

\smallskip

We use the symbols $\ns$, $\nsii$, $\nsiii$, $\ldots$ to denote nested sequents. For a nested sequent of the form $\Gamma \sar \Delta, (\forb_{1})[\ns_{1}], \ldots, (\forb_{n})[\ns_{n}]$, we refer to $\Gamma$ as the \emph{antecedent} and to $\Delta, (\forb_{1})[\ns_{1}], \ldots, (\forb_{n})[\ns_{n}]$ the \emph{consequent}. We call $(\fd)[\ns]$ a \emph{forward nesting}, $(\bd)[\ns]$ a \emph{backward nesting}, and define a \emph{nesting} to be either a forward or backward nesting. We let $\forb \in \set{\fd,\bd}$, and define $\conv{\forb} := \fd$ if $\forb = \bd$ and $\conv{\forb} := \bd$ if $\forb = \fd$. %A \emph{modal (nested) sequent} is a nested sequent using only formulae from the language $\langm$ and forward nestings. 
Moreover, we let $\xdia \in \set{\wdia,\bdia}$ and $\xbox \in \set{\wbox,\bbox}$ such that $\xdia = \wdia$ and $\xbox = \wbox$ if $\forb = \fd$, and $\xdia = \bdia$ and $\xbox = \bbox$ if $\forb = \bd$. %We let $\seqset$ denote the set of all Gentzen sequents.

\begin{remark} Although $\Gamma$, $\Delta$, etc. are treated as \emph{sets} in Gentzen sequents, the collection of nestings $(\forb_{1})[\ns_{1}], \ldots, (\forb_{n})[\ns_{n}]$ occurring in the consequent of a nested sequent is assumed to form a \emph{multiset}. This distinction will be useful later on for detecting loops during proof-search, which is required for termination.
\end{remark}

%Additionally, we call $\Gamma$ a \emph{left-collection} \iffi it is the antecedent of some nested sequent, and we call $\Delta, (\charx_{1})[\ns_{1}], \ldots, (\charx_{n})[\ns_{n}]$ a \emph{right-collection} \iffi it is the consequent of some nested sequent; this terminology will be helpful below. Given a nested sequent $\ns$, we let $\rcol$ denote its consequent; i.e., we use $\rcol$, $\rcolii$, $\nsiii^{+}$, $\ldots$ to denote right-collections. We let $\col$, $\colii$, $\coliii$, $\ldots$ denote either left- or right-collections.

A \emph{component} of $\ns = \Gamma \sar \Delta, (\forb_{1})[\nsii_{1}], \ldots, (\forb_{n})[\nsii_{n}]$ is defined to be a Gentzen sequent appearing therein, that is, a component of $\ns$ is an element of the multiset $c(\ns)$, where $c$ is defined as follows and $\uplus$ denotes the multiset union:
$$
c( \Gamma \sar \Delta, (\forb_{1})[\nsii_{1}], \ldots, (\forb_{n})[\nsii_{n}]) := \{ \Gamma \sar \Delta \} \uplus \!\! \biguplus_{i \in [n]} \!\! c(\nsii_{i}).
$$
We let $\names := \set{w,u,v,\ldots}$ be a denumerable set of pairwise distinct labels, called \emph{names}. Given a nested sequent $\ns$ of the form $\Gamma \sar \Delta, (\forb_{1})[\nsii_{1}], \ldots, (\forb_{n})[\nsii_{n}]$, we call $\Gamma \sar \Delta$ the \emph{root} of $\ns$ and assume that every component of $\ns$ is assigned a unique name from $\names$. We sometimes refer to a component as a \emph{$w$-component} if $w$ is its name. The incorporation of names in our nested sequents is crucial for extracting models from failed proof-search. We define $\names(\ns)$ to be the set of all names assigned to components in $\ns$. Nested sequents are special kinds of labeled trees, called \emph{seq-trees}.

\begin{definition}[Seq-Tree] %We define a \emph{labeled tree} to be a triple $\ltr = (\trv,\tre,\trl)$ such that $\trv$ is a non-empty set of \emph{vertices}, $\tre \ \subseteq \trv \times \trv$ is a finite set of \emph{edges}, $\trl$ is a function mapping vertices and edges to a set of labels, and $(\trv,\tre)$ is a directed tree, i.e., there exists a unique vertex $\rt \in \trv$ called the \emph{root} such that for any other vertex $u \in \trv$, a single path exists from $\rt$ to $u$ in $\ltr$. A \emph{seq-tree} is defined to be a labeled tree $\ltr = (\trv,\tre,\trl)$ such that $\trl : \trv \cup \tre \to \seqset \cup \set{\fd,\bd}$, ${\trl}{\restriction_{\trv}} : \trv \to \seqset$, and ${\trl}{\restriction_{\tre}} : \tre \to \set{\fd,\bd}$.
Let $\ns =  \Gamma \sar \Delta, (\forb_{1})[\nsii_{1}], \ldots, (\forb_{n})[\nsii_{n}]$ be a nested sequent with $\rt$ the name of the root and $u_{i}$ the name of the root of $\nsii_{i}$ for $i \in [n]$. We define its corresponding \emph{seq-tree} $\tr{\ns} = (\trv,\tre,\trl)$ as follows:
\begin{description}

\item[$\bullet$] $\trv := \{w\} \cup V_{1} \cup \cdots \cup \trv_{n}$

\item[$\bullet$] $\tre := \{(w,u_{i}) \ | \ i \in [n]\} \cup \tre_{1} \cup \cdots \cup \tre_{n}$

\item[$\bullet$] $\trl := \set{(w,\Gamma \sar \Delta), (w,u_{i},\forb_{i})} \cup \trl_{1} \cup \cdots \cup \trl_{n}$

\end{description}
%$\trl(w) = \Gamma \sar \Delta$, $\trl(w,u_{i}) = \forb_{i}$, and 
such that $\tr{\nsii_{i}} = (\trv_{i}, \tre_{i}, \trl_{i})$ for $i \in [n]$. We write $u \in \tr{\ns}$ and $(u,v) \in \tr{\ns}$ as a shorthand for $u \in \trv$ and $(u,v) \in \tre$, respectively.
\end{definition}

A nested sequent $\ns$ and its associated seq-tree $\tr{\ns}$ are simply two representations of the same underlying structure, so we will identify them interchangeably throughout the paper. We remark that we will apply standard terminology for trees when discussing seq-trees and nested sequents, e.g., leaf, successor, ancestor, etc. (see \cite[Chapter 11]{Ros19}). We note that the third component $\trl$ of a seq-tree $\tr{\ns} = (\trv,\tre,\trl)$ is effectively a function that maps names (i.e., vertices) to Gentzen sequents and ordered pairs of names (i.e., edges) to $\fd$ or $\bd$. Therefore, we will often write $\trl(w) = \Gamma \sar \Delta$ for $(w, \Gamma \sar \Delta) \in \trl$ and $\trl(w,u) = \forb$ for $(w, u, \forb) \in \trl$.

The \emph{depth} of a seq-tree or nested sequent is equal to the number of nodes along a maximal path from the root to a leaf. Let $w$ be the root of a nested sequent $\ns$. We recursively define the \emph{depth of components} in $\ns$ as follows: (1) $w \in \depk{\ns}{0}$, and (2) if $(u,v) \in \tr{\ns}$ and $u \in \depk{\ns}{k}$, then $v \in \depk{\ns}{k+1}$. We define the \emph{modal depth} of a nested sequent $\ns$ to be the sum of the modal depth of each component with its depth in $\ns$, which is formally defined as follows:
$$
\mdep{\ns} := \max\set{\mdep{\Gamma, \Delta} + k \mid u \in \tr{\ns}, \trl(u) = \Gamma \sar \Delta, u \in \depk{\ns}{k}}.
$$

We define the \emph{input formulae} and \emph{output formulae} of a seq-tree $\ltr = (\trv,\tre,\trl)$ to be all formulae occurring in antecedents of vertex labels (shown below left) and consequents of vertex labels (shown below right), respectively.
$$
\infs{\ltr} := \!\!\!\!\!\!\!\!\!\!\!\! \bigcup_{w \in \trv, \ \trl(w) =  \Gamma \sar \Delta} \!\!\!\!\!\!\!\!\!\!\!\! \Gamma
\qquad
\outfs{\ltr} := \!\!\!\!\!\!\!\!\!\!\!\! \bigcup_{w \in \trv, \ \trl(w) = \Gamma \sar \Delta} \!\!\!\!\!\!\!\!\!\!\!\! \Delta
$$
For $w \in \trv$ with $\trl(w) = \Gamma \sar \Delta$, $\infs{w,\ltr} := \Gamma$ and $\outfs{w,\ltr} := \Delta$. % to be all input and output formulae in that component, respectively.

%Observe that every nested sequent encodes a labeled tree whose nodes are (named) Gentzen sequents.

\begin{figure*}[t]
\noindent

\begin{center}
\begin{tabular}{c c c} % @{\hskip 1em} c}
\AxiomC{$\phantom{\ns}$}
\RightLabel{$\id$}
\UnaryInfC{$\ns \sbl \Gamma, p \sar p, \Delta \sbr$}
\DisplayProof

&

\AxiomC{$\phantom{\ns}$}
\RightLabel{$\botl$}
\UnaryInfC{$\ns \sbl \Gamma, \bot \sar \Delta \sbr$}
\DisplayProof

&

\AxiomC{$\ns \sbl \Gamma, A \sar \Delta \sbr$}
\AxiomC{$\ns \sbl \Gamma, B \sar \Delta \sbr$}
\RightLabel{$\disl$}
\BinaryInfC{$\ns \sbl \Gamma, A \lor B \sar \Delta \sbr$}
\DisplayProof
\end{tabular}
\end{center}

\begin{center}
\begin{tabular}{c c} % @{\hskip 1em} c}
\AxiomC{$\ns \sbl \Gamma \sar A, B, \Delta \sbr$}
\RightLabel{$\disr$}
\UnaryInfC{$\ns \sbl \Gamma \sar A \lor B, \Delta \sbr$}
\DisplayProof

&

\AxiomC{$\ns \sbl \Gamma, A, B \sar \Delta \sbr$}
\RightLabel{$\conl$}
\UnaryInfC{$\ns \sbl \Gamma, A \land B \sar \Delta \sbr$}
\DisplayProof
\end{tabular}
\end{center}
\begin{center}
\begin{tabular}{c c} % @{\hskip 1em} c}
\AxiomC{$\ns \sbl \Gamma \sar A, \Delta \sbr$}
\AxiomC{$\ns \sbl \Gamma \sar B, \Delta \sbr$}
\RightLabel{$\conr$}
\BinaryInfC{$\ns \sbl \Gamma \sar A \land B, \Delta \sbr$}
\DisplayProof

&

\AxiomC{$\ns \sbl \Gamma, A \iimp B \sar A, \Delta \sbr$}
\AxiomC{$\ns \sbl \Gamma, B \sar \Delta \sbr$}
\RightLabel{$\iimpl$}
\BinaryInfC{$\ns \sbl \Gamma, A \iimp B \sar \Delta \sbr$}
\DisplayProof
\end{tabular}
\end{center}
\begin{center}
\begin{tabular}{c c} % @{\hskip 1em} c}
\AxiomC{$\ns^{\downarrow} \sbl \Gamma, A \sar B \sbr$}
\RightLabel{$\iimpr$}
\UnaryInfC{$\ns \sbl \Gamma \sar A \iimp B, \Delta \sbr$}
\DisplayProof

&

\AxiomC{$\ns \sbl \Gamma \sar \Delta, (\forb) [ A \sar \ ] \sbr$}
\RightLabel{$\xdial$}
\UnaryInfC{$\ns \sbl \Gamma, \xdia A \sar \Delta \sbr$}
\DisplayProof
\end{tabular}
\end{center}
\begin{center}
\begin{tabular}{c c}
\AxiomC{$\ns^{\da} \sbl \Gamma \sar (\forb) [ \ \sar A ] \sbr$}
\RightLabel{$\xboxr$}
\UnaryInfC{$\ns \sbl \Gamma \sar \xbox A, \Delta \sbr$}
\DisplayProof

&

\AxiomC{$\ns \sbl \Gamma \sar \xdia A, \Delta \sbr_{w} \sbl \Sigma \sar A, \Pi \sbr_{u} $}
\RightLabel{$\xdiar^{\dag(\cons)}$}
\UnaryInfC{$\ns \sbl \Gamma \sar \xdia A, \Delta \sbr_{w} \sbl \Sigma \sar \Pi \sbr_{u} $}
\DisplayProof
\end{tabular}
\end{center}
\begin{center}
\begin{tabular}{c c}
\AxiomC{$\ns \sbl \Gamma, \xbox A \sar \Delta \sbr_{w} \sbl \Sigma, A \sar \Pi \sbr_{u}$}
\RightLabel{$\xboxl^{\dag(\cons)}$}
\UnaryInfC{$\ns \sbl \Gamma, \xbox A \sar \Delta \sbr_{w} \sbl \Sigma \sar \Pi \sbr_{u}$}
\DisplayProof

&

\AxiomC{$\ns \sbl \Gamma \sar \Delta, (\fd)[ \  \sar \ ] \sbr$}
\RightLabel{$\ddr$}
\UnaryInfC{$\ns \sbl \Gamma \sar \Delta \sbr$}
\DisplayProof
\end{tabular}
\end{center}

%%%%%%%%
\iffalse
\begin{center}
\begin{tabular}{c c}
\AxiomC{$\ns \sbl \Gamma \sar \xdia A, \Delta, (\forb)[\Sigma, A \sar \Pi] \sbr$}
\RightLabel{$\xdiari$}
\UnaryInfC{$\ns \sbl \Gamma \sar \xdia A, \Delta, (\forb)[\Sigma \sar \Pi] \sbr$}
\DisplayProof

&

\AxiomC{$\ns \sbl \Gamma \sar A, \Delta, (\conv{\forb})[\Sigma, \xdia A \sar \Pi] \sbr$}
\RightLabel{$\xdiarii$}
\UnaryInfC{$\ns \sbl \Gamma \sar \Delta, (\conv{\forb})[\Sigma, \xdia A \sar \Pi] \sbr$}
\DisplayProof
\end{tabular}
\end{center}

\begin{center}
\begin{tabular}{c c}
\AxiomC{$\ns \sbl \Gamma, \xbox A \sar \Delta, (\forb)[\Sigma, A \sar \Pi] \sbr$}
\RightLabel{$\xboxli$}
\UnaryInfC{$\ns \sbl \Gamma, \xbox A \sar \Delta, (\forb)[\Sigma \sar \Pi] \sbr$}
\DisplayProof

&

\AxiomC{$\ns \sbl \Gamma, A \sar \Delta, (\conv{\forb})[\Sigma, \xbox A \sar \Pi] \sbr$}
\RightLabel{$\xboxlii$}
\UnaryInfC{$\ns \sbl \Gamma \sar \Delta, (\conv{\forb})[\Sigma, \xbox A \sar \Pi] \sbr$}
\DisplayProof
\end{tabular}
\end{center}
\fi
%%%%%%%%

\smallskip

\begin{flushleft}
\textbf{Side Condition:} $\dag(\cons) := w \rable_{\forb}^{\cons} u$
\end{flushleft}

\caption{Nested Sequent Rules.\label{fig:nested-calculus}} %The side condition on $\wdiar$, $\bdiar$, $\wboxl$, and $\bboxl$ is $\dag(w,u,\X) := w \rableosf^{\X} u$
\end{figure*}

A \emph{context} is a nested sequent with \emph{holes}; a hole $\nshole$ takes the place of a Gentzen sequent in a nested sequent (cf.~\cite{Bru09,GorPosTiu11}). For example, $\ns\nshole = A,B \sar (\fd)[\nshole, (\bd)[C \sar D]]$ is a context with one hole and $\nsii\nshole\nshole = \nshole, (\fd)[\nshole], (\bd)[A \sar B]$ is a context with two holes. In a context $\ns\nshole\cdots\nshole$ with $n$ holes, we can substitute the nested sequents $\nsii_{1}$, $\ldots$, $\nsii_{n}$ for each hole, respectively, to obtain a nested sequent $\ns\{\nsii_{1}\}\ldots\{\nsii_{n}\}$. For example, if we substitute the nested sequent $\nsii = E \sar F, (\fd)[G \sar ]$ in the context $\ns\nshole$ above, we obtain the nested sequent:
$$
\ns\set{\nsii} = A,B \sar (\fd)[E \sar F, (\fd)[G \sar ], (\bd)[C \sar D]]
$$
We sometimes write $\ns\{\nsii_{1}\}_{w_{1}}\ldots\{\nsii_{n}\}_{w_{n}}$ to indicate that the root of $\nsii_{i}$ is $w_{i}$ for $i \in [n]$. As seen later on, contexts and holes are helpful for formulating our inference rules.

\begin{definition}[Extended Semantics]\label{def:sequent-semantics} Let $\tr{\ns} = (\trv,\tre,\trl)$ be the seq-tree of $\ns$ with $\rt$ the name of its root and let $M = (W, \leq, R, V)$ be a model. An \emph{$M$-interpretation} is a function $\mint : \trv \rightarrow W$. We say $\ltr$ is satisfied on $M$ with $\mint$, written $M, \mint \sat \ltr$, \iffi if conditions (1) and (2) hold, then (3) holds:
\begin{description}

\item[$(1)$] for each $u \tre v$ with $\trl(u,v) = \fd$, $\mint(u)R\mint(v)$;

\item[$(2)$] for each $u \tre v$ with $\trl(u,v) = \bd$, $\mint(v)R\mint(u)$;

\item[$(3)$] for some $v \in \trv$ with $\trl(v) = \Gamma \sar \Delta$, $M, \mint(v) \sat \bigwedge \Gamma \iimp \bigvee \Delta$.

\end{description}
We say $\tr{\ns}$ is \emph{$\fcons$-valid} \iffi for every model $M$ based on a frame in $\fcons$ and $M$-interpretation $\mint$, $M, \mint \sat \tr{\ns}$; otherwise, $\tr{\ns}$ is \emph{$\frames$-invalid}. For a nested sequent $\ns$, we define $M, \mint \sat \ns$ \iffi $M, \mint \sat \tr{\ns}$, and say $\ns$ is $\fcons$-valid \iffi $\tr{\ns}$ is $\fcons$-valid.
\end{definition}

%\begin{definition}[Formula Interpretation] We define the formula interpretation $\fm(\ns)$ of a nested sequent $\ns$ % = \Gamma \sar \Delta, [\nsii_1],\dots,[\nsii_n]$ 
%accordingly:
%$$
%\fm(\Gamma \sar \Delta, (\charx_{1})[\nsii_1], \ldots, (\charx_{n})[\nsii_n]) := \bigwedge \Gamma \iimp (\bigvee \Delta \lor \bigvee_{i=1}^{n} [\charx_{i}] \,\fm(\nsii_{i}) )
%$$
%where $\bigwedge \emptyset := \top$ and $\bigvee \emptyset := \bot$, as usual.
%\end{definition}

%\begin{proposition}\label{prop:equiv-model-form-interp}
%Let  $\ns = \Gamma \sar \Delta,(\charx_{1})[\nsii_1], \ldots, (\charx_{n})[\nsii_n]$ be a nested sequent with $w$ the name of the root, $M = (W, \leq, \{R_{\charx} \ | \ \charx \in \albet\}, V)$ be a model, and $\mint$ be an $M$-interpretation. Then, $M, \mint \vDash \ns$ \iffi $M, \mint(w) \vDash \fm(\ns)$.
%\end{proposition}

The rules used in our nested sequent systems are displayed in \fig~\ref{fig:nested-calculus}. The \emph{initial rules} are $\id$ and $\botl$, which are used to generate the axioms of the system, and the remaining rules are \emph{logical rules} forming left and right pairs. We remark that the modal rules $\xdial$, $\xboxl$, $\xdiar$, and $\xboxr$ take $\forb$ as a parameter, and so, each instantiates to either a forward or backward modal rule: $\xdial \in \set{\wdial,\bdial}$, $\xboxl \in \set{\wboxl,\bboxl}$, $\xdiar \in \set{\wdiar,\bdiar}$, $\xboxr \in \set{\wboxr,\bboxr}$. For completeness, explicit presentations of these rules have been included in the appendix. See Example~\ref{ex:propagation-graph-path} for a $\bboxl$ application.

As usual, we refer to the formula(e) explicitly introduced in the conclusion of a rule as \emph{principal}, and refer to those displayed in the premise(s), used to infer the principal formula, as \emph{auxiliary}. For example, $\xbox A$ is principal in $\xboxl$ while $A$ is auxiliary. We let $\ns^{\da}\nshole\cdots\nshole$ denote the context obtained from $\ns\nshole\cdots\nshole$ by deleting all consequent formulae from all components; this operation is used in the $\iimpr$ and $\xboxr$ rules to ensure soundness. $\xdial$ and $\xboxr$ bottom-up `unpack' a modal formula, which introduces a new component that we assume is named with a fresh name from $\names$.

The $\xdiar$ and $\xboxl$ rules are special logical rules called \emph{propagation rules} (cf.~\cite{CiaLyoRamTiu21,GorPosTiu11}). Propagation rules enable formulae to be (bottom-up) propagated to certain components in a nested sequent. These rules take $\cons$ as a parameter determining \emph{how} formulae are to be propagated via the side condition $w \rable_{\forb}^{\cons} u$, defined below.

\begin{definition}\label{def:closure} Let $\ns$ be a nested sequent and $\tr{\ns} = (\trv,\tre,\trl)$ with $w,u \in \trv$. We define the \emph{$\cons$-closure} of $\tr{\ns}$, written $\tr{\ns}^{\cons} = (\trv_{\cons},\tre_{\cons},\trl_{\cons})$, to be the minimal graph such that $\trv_{\cons} := \trv$, $\tre \subseteq \tre_{\cons}$, $\trl \subseteq \trl_{\cons}$, and the following are satisfied, for all $w,u \in \trv$:
\begin{description}[leftmargin=!, labelwidth=0em, labelsep=.25em]

\item[$(1)$] if $w \tre u$, then $(w, u), (u, w) \in \ \tre_{\cons}$;

\item[$(2)$] if $w \tre u$ and $\trl(w,u) = \forb$, then $\trl_{\cons}(w,u) = \forb$ and $\trl_{\cons}(u,w) = \conv{\forb}$;

%\item[$(3)$] if $(w,u) \in \tre$ and $\trl(w,u) = \forb$, then $\trl_{\cons}(u,w) = \conv{\forb}$;

\item[$(3)$] if $\T \in \cons$, then there exist two loops $(w,w), (w,w) \in \tre_{\cons}$ with $\trl_{\cons}(w,w) = \fd$ and $\trl_{\cons}(w,w) = \bd$;

\item[$(4)$] if $\B \in \cons$, $w \tre u$, and $\trl(w,u) = \forb$, then $(w, u), (u, w) \in \ \tre_{\cons}$ with $\trl_{\cons}(w,u) = \conv{\forb}$ and $\trl_{\cons}(u,w) = \forb$.

\end{description}
We define $\tr{\ns} \models w \rable_{\forb}^{\cons} u$ \iffi $(w,u) \in \tre_{\cons}$ and $\trl_{\cons}(w,u) = \forb$. For simplicity, we write $w \rable_{\forb}^{\cons} u$ when $\tr{\ns}$ is clear from the context. % or $\trl_{\cons}(u,w) = \conv{\forb}$.
\end{definition}

If a propagation rule is applied top-down (bottom-up), we assume that $\dag(\cons) =$ `$\tr{\ns} \models w \rable_{\forb}^{\cons} u$' with $\ns$ the premise (conclusion, resp.). We provide an example to make the application of such rules clear.

\begin{example}\label{ex:propagation-graph-path} Let $\ns = \Gamma, \bbox p, p \sar \Delta, (\fd)[\Sigma, p \sar \Pi]$ with $w$ the name of the root, $u$ the name of $\Sigma, p \sar \Pi$, and suppose $\B, \T \in \cons$. A graphical depiction of the $\cons$-closure $\tr{\ns}^{\cons}$ is given below:
%\resizebox{0.97\linewidth}{!}{%
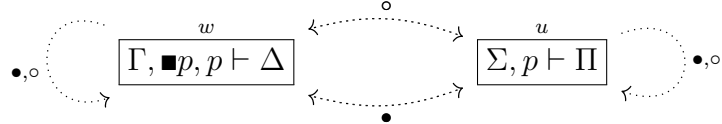
\[
\begin{tikzcd}[column sep=5em]
  \overset{w}{\boxed{\Gamma, \bbox p, p \sar \Delta}}
    \arrow[r, bend left=18, dotted, "\fd"]
    \arrow[r, bend right=18, dotted, swap, "\bd"]
%   \arrow[loop left, out=180, in=225, looseness=3, dotted, swap, "\fd"]
    \arrow[loop left, out=160, in=200, looseness=3, dotted, swap, "{\bd,\fd}"]
  &
  \overset{u}{\boxed{\Sigma, p \sar \Pi}}
    \arrow[l, bend left=18, dotted, "\bd"]
    \arrow[l, bend right=18, dotted, swap, "\fd"]
%   \arrow[loop right, in=0, out=-45, looseness=2.5, dotted, "\fd"]
    \arrow[loop right, out=20, in=340, looseness=4, dotted, "{\bd,\fd}"]
\end{tikzcd}
\]
We indicate the double edges arising from $\B$ as double arrows and the double loops arising from $\T$ are indicated as loops with both labels $\fd$ and $\bd$. One can see that $\tr{\ns} \models w \rable_{\bd}^{\cons} w$ and $\tr{\ns} \models w \rable_{\bd}^{\cons} u$, so we are permitted to (top-down) apply $\bboxl$ to $\ns$ to infer both $\Gamma, \bbox p \sar \Delta, (\fd)[\Sigma, p \sar \Pi]$ and $\Gamma, \bbox p, p \sar \Delta, (\fd)[\Sigma \sar \Pi]$, respectively.
\end{example}

\begin{definition}[Nested Sequent System] Let $\cons \subseteq \set{\T,\B,\D}$. The nested sequent system $\nikt$ is the set of all rules, except $\ddr$, in Figure~\ref{fig:nested-calculus} and where the side condition $\dag(\emptyset)$ is imposed on $\xdiar$ and $\xboxl$. For $\cons \neq \emptyset$, we define $\niktc$ to be the system obtained by extending $\nikt$ such that (1) $\ddr \in \niktc$ \iffi $\D \in \cons$ and (2) the side condition $\dag(\cons)$ is imposed on $\xdiar$ and $\xboxl$.
\end{definition}

A \emph{derivation} in a nested sequent calculus of $\ns$ is a (potentially infinite) tree whose nodes are labeled with nested sequents such that: (1) The root is labeled with $\ns$, and (2) Every parent node is the conclusion of a rule of the calculus with its children the corresponding premises. We define a \emph{branch} $\branch = \ns_{0}, \ns_{1}, \ldots, \ns_{n}, \ldots$ to be a maximal path of nested sequents in a derivation such that $\ns_{0}$ is the conclusion of the derivation and each nested sequent $\ns_{i+1}$ (if it exists) is a child of $\ns_{i}$.

A \emph{proof} is a finite derivation where all leaves are instances of initial rules. We use $\prf$ and annotated versions thereof to denote both derivations and proofs with the context differentiating the usage. If a proof of a nested sequent $\ns$ exists in a nested sequent calculus, then we write $\calc\proves\ns$ to indicate this. The \emph{height} of a derivation is defined as the number of nested sequents along a maximal branch in the derivation, which may be infinite if the derivation is infinite.

\begin{example} We give a proof of $(\bdia p \iimp \wdia p) \land (p \iimp \wbox \wdia p)$, which is provable in $\calc$ when $\B \in \cons$. Suppose $\B \in \cons$, and observe that $\tr{\ns} \models w \rable_{\fd}^{\cons} u$ and $\tr{\nsii} \models w \rable_{\fd}^{\cons} u$ with $\ns = \ \sar \wdia p, (\bd)[p \sar p]$, $\nsii = \ p \sar p, (\fd)[ \sar \wdia p]$, $w$ the name of the root of each nested sequent, and $u$ the name of the only non-root component. Therefore, $\wdiar$ may be applied in the left and right branches.
\begin{center}
%\resizebox{\columnwidth}{!}{
%\begin{tabular}{@{\hskip -.1cm} c}
\AxiomC{}%$\phantom{\ns}$}
\RightLabel{$\id$}

\UnaryInfC{$\sar \wdia p, (\bd)[p \sar p]$}
\RightLabel{$\wdiar$}
\UnaryInfC{$\sar \wdia p, (\bd)[p \sar  ]$}
\RightLabel{$\bdial$}
\UnaryInfC{$\bdia p \sar \wdia p$}
\RightLabel{$\iimpr$}
\UnaryInfC{$\sar \bdia p \iimp \wdia p$}

\AxiomC{}%$\phantom{\ns}$}
\RightLabel{$\id$}

\UnaryInfC{$p \sar p, (\fd)[ \sar \wdia p]$}
\RightLabel{$\wdiar$}
\UnaryInfC{$p \sar (\fd)[ \sar \wdia p]$}
\RightLabel{$\wboxr$}
\UnaryInfC{$p \sar \wbox \wdia p$}
\RightLabel{$\iimpr$}
\UnaryInfC{$\sar p \iimp \wbox \wdia p$}

\RightLabel{$\conr$}
\BinaryInfC{$\sar (\bdia p \iimp \wdia p) \land (p \iimp \wbox \wdia p)$}
\DisplayProof
%}
%\end{tabular}
\end{center}
\end{example}

%\subsection{Soundness and Completeness}

The nested systems presented in this section are variants of systems that appear in the literature. In particular, \cite{Lyo22,Lyo25} introduce single-conclusioned nested calculi for intuitionistic grammar logics~\cite{Lyo21b}, a class of logics that includes $\iktc$ for $\cons \subseteq \set{\T,\B,\D}$. Each system 
$\niktc$ considered here is obtained as a multi-conclusioned version of one of these earlier calculi. Given this close relationship, it is natural that each nested system $\niktc$ is sound and complete.

%The nested systems discussed in this section are variants of nested systems that already exist in the literature. In~\cite{Lyo22,Lyo25}, single-conclusioned nested systems were provided for intuitionistic grammar logics~\cite{Lyo21b}, which is a class of logics including $\iktc$ for $\cons \subseteq \set{\T,\B,\D}$. Each system $\niktc$ is a multi-conclusioned version of one of the aforementioned systems. Due to the connection between the nested systems presented in this section and the existing ones in the literature, it is unsurprising that each nested system $\calc$ is sound and complete. 

%Each nested system $\calc$ can be proven sound in the usual way; a proof can be found in the appendix.

\begin{theorem}[Soundness]\label{thm:soundness-nested} If $\calc \proves \ns$, then $\ns$ is $\fcons$-valid.
\end{theorem}

We note that cut-free completeness is a corollary of the decidability proof given in the following section. This is because if $\calc \not\proves A$, then our decision procedure yields a model $M$ based on a frame in $\fcons$ falsifying $A$, which is equivalent to establishing completeness. We therefore have:

\begin{theorem}[Completeness]\label{thm:completeness-nested}  If $\fcons \sat A$, then $\calc \proves A$.
\end{theorem}

%\paragraph{Separation Property.} Our nested sequent systems also satisfy the \emph{separation property}, originally introduced in the context of nested calculi for classical tense logics~\cite{GorPosTiu11}. A tense calculus~$\niktsfour$ or $\niktsfour$ has the separation property if and only if, every derivation of $A \in \langm$ is a derivation of $A$ in~$\nikfour$ or $\nisfour$ (resp.). It is straightforward to verify that $\niktsfour$ or $\niktsfour$ satisfy the separation property: if the conclusion of a derivation is a modal formula, then none of the rules $\set{\bdial,\bdiar,\bboxl,\bboxr}$ introducing tense modalities can ever be applied in the derivation. This property is significant in our setting because is means that a proof-search algorithm for $\niktsfour$ or $\niktsfour$ coincides with a proof-search algorithm for $\nikfour$ or $\nisfour$ (resp.) when the input is a modal formula, since the rules for tense modalities are rendered irrelevant. 

%%%Tranding Rules
\section{Proof-Search and Decidability}\label{sec:decid}

%%%%%ALG%%%%%%%
%%%%%%%%%%%%%%%

\begin{algorithm}[t]%\label{alg:Proof-search}
\KwIn{A Nested Sequent: $\ns$}
\KwOut{A Boolean: $\true$, $\false$}

\If{$A \in \infs{\tr{\ns}} \cap \outfs{\tr{\ns}}$ or $\bot \in \infs{\tr{\ns}}$}
     {\Return $\true$;}

\If{$\tr{\ns}$ is stable or a repeat}
     {\Return $\false$;}

\If{$A \lor B \in \infs{w,\tr{\ns}}$, but $A, B \not\in \infs{w,\tr{\ns}}$}
{
    Let $\ns_{1} := \ns \triangleleft_{w} (A \sar \ )$ and $\ns_{2} := \ns \triangleleft_{w} (B \sar \ )$;\\
    \Return $\prove(\ns_{1}) \ \&\& \ \prove(\ns_{2})$;
}

\If{$A \lor B \in \outfs{w,\tr{\ns}}$, but $\set{A,B} \not\subseteq \outfs{w,\tr{\ns}}$}
{
    Let $\ns' := \ns \triangleleft_{w} ( \ \sar A,B )$; \Return $\prove(\ns')$;
}

\If{$A \land B \in \infs{w,\tr{\ns}}$, but $\set{A,B} \not\subseteq \infs{w,\tr{\ns}}$}
{
    Let $\ns' := \ns \triangleleft_{w} (A,B \sar \ )$; \Return $\prove(\ns')$;
}

\If{$A \land B \in \outfs{w,\tr{\ns}}$, but $A, B \not\in \outfs{w,\tr{\ns}}$}
{
    Let $\ns_{1} := \ns \triangleleft_{w} ( \ \sar A)$ and $\ns_{2} := \ns \triangleleft_{w} ( \ \sar B)$;\\
    \Return $\prove(\ns_{1}) \ \&\& \ \prove(\ns_{2})$;
}

\If{$A \iimp B \in \infs{w,\tr{\ns}}$, but $A \not\in \outfs{w,\tr{\ns}}$ and $B \not\in \infs{w,\tr{\ns}}$}
{
    Let $\ns_{1} := \ns \triangleleft_{w} ( \ \sar A)$ and $\ns_{2} := \ns \triangleleft_{w} (B \sar \ )$;\\
    \Return $\prove(\ns_{1}) \ \&\& \ \prove(\ns_{2})$;
}

\caption{$\prove$ (Part 1)}
\end{algorithm}
%%%%%%%%%%%%%%
%%%%%%%%%%%%%%

%%%%%%ALG%%%%%%%%
%%%%%%%%%%%%%%
\setcounter{algocf}{0}
\begin{algorithm}[t]%\label{alg:Proof-search-ii}
\setcounter{AlgoLine}{24}

\If{$\xdia A \in \infs{w,\tr{\ns}}$, and no $u \in \tr{\ns}$ exists with %$\prgr{\tr{\ns}} \models 
$w \rableosx u$ and $A \in \infs{v,\tr{\ns}}$}
{
    Let $\ns' := \ns \triangleleft_{w} ( \ \sar (\forb)[A \sar \ ])$ with $u$ the fresh name of the new component; \Return $\prove(\ns')$;
}

\If{$\xdia A \in \outfs{w,\tr{\ns}}$, %for some $u \in \tr{\ns}$, 
$w \rableosx u$, and $A \not\in \outfs{u,\tr{\ns}}$}
{
    Let $\ns' := \ns \triangleleft_{u} ( \ \sar A)$; \Return $\prove(\ns')$;
}

\If{$\xbox A \in \infs{w,\tr{\ns}}$, $w \rableosx u$, and $A \not\in \infs{u,\tr{\ns}}$}
{
    Let $\ns' := \ns \triangleleft_{u} (A \sar \ )$; \Return $\prove(\ns')$;
}

\If{$\D \in \cons$, $w \in \depk{\ns}{k}$, $k \leq \mdep{\nsinp}$, and no $u \in \tr{\ns}$ exists such that $w \tre u$ and $\trl(w,u) = \fd$}
{
    Let $\ns' := \ns \triangleleft_{w} (\fd)[ \  \sar \ ]$ such that $u$ is the fresh name of the new component;\\
    \Return $\prove(\ns')$;
}

\If{$\tr{\ns}$ is saturated}
{
Let $A_{i} \iimp B_{i} \in \outfs{w_{i},\tr{\ns}}$ for $1 \leq i \leq n$ be all output $\iimp$-formulae in $\ns$;\\
Set $\ns_{i} := \ns^{\da} \triangleleft_{w} (A_{i} \sar B_{i})$;\\
Let $[\forb_{i}] C_{i} \in \outfs{u_{i},\tr{\ns}}$ for $n + 1 \leq i \leq n + k$ be all output $\wbox$- or $\bbox$-formulae in $\ns$;\\ % such that $u_{i}$ is not blocked, and 
Set $\ns_{i} := \ns^{\da} \triangleleft_{w} ( \ \sar  (\fd)[\ \sar A])$ with $u_{i}'$ a fresh name of the new component;\\
\Return $\prove(\ns_{1}) \ \| \cdots \| \ \prove(\ns_{n+k})$;
}

\caption{$\prove$ (Part 2)}
\end{algorithm}
%%%%%%%%%%%%%%
%%%%%%%%%%%%%%

We now present a proof-search algorithm deciding $\iktc$ for all $\cons \subseteq \set{\T,\B,\D}$. The first section details the internal workings of the algorithm and introduces computation trees, the structures generated during proof search that encode multiple derivations of the input. The second section establishes bounds on the depth of nested sequents and shows how branches of computation trees are finitely bounded via a novel loop-checking mechanism based on homomorphic mappings between nested sequents. %, thereby ensuring termination. 
The final section describes how counter-models are extracted from computation trees when proof-search fails.

\subsection{Computation Trees and Proof Extraction}\label{subsec:comp-trees}

We now introduce two properties, \emph{$\cons$-saturated} and \emph{$\cons$-stable}, that identify redundant rule applications during proof-search and help ensure termination. In essence, a nested sequent $\ns$ is $\cons$-saturated if the only non-redundantly applicable rules are the non-invertible rules $\iimpr$ and $\xboxr$. When such sequents are encountered, the algorithm applies both rules bottom-up in all possible combinations, corresponding to disjunctive branching; that is, it suffices for one bottom-up application to succeed in order to construct a proof of $\ns$. A nested sequent is $\cons$-stable if no rules are non-redundantly applicable to it, indicating that proof-search on the corresponding branch may safely terminate.

%Let us now define our saturation conditions. We remark that when $\D \in \cons$, i.e., $\ddr \in \calc$

\begin{remark}
For the remainder of Section~\ref{sec:decid}, we fix the symbol $\nsinp$ to denote the input to $\prove$, unless stated otherwise.
\end{remark}

\begin{figure*}[t]
\centering

\AxiomC{$\ns^{\da}\set{\Gamma_{i} \sar (\forb_{i,j})[\emptyset \sar A_{i,j}]}_{w_{i}} \mid i \in [n], j \in [k_{i}]$ \quad
$\ns^{\da}\set{\Gamma_{i}, C_{i,j} \sar D_{i,j}}_{w_{i}} \mid i \in [n], j \in [n_{i}]$}
%$\set{\ns^{\da}\set{\Gamma_{i} \sar (\fd)[\emptyset \sar A_{i,j}]}_{w_{i}} \mid i \in [m], j \in [k_{i}]}$}
%\AxiomC{$\set{\ns^{\da}\set{\Gamma_{i}, B_{i,j} \sar C_{i,j}}_{w_{i}} \mid i \in [m], j \in [n_{i}]}$}
\RightLabel{$\dbr$}
%\dashedLine
\UnaryInfC{$\ns\set{\Gamma_{1} \sar \Pi_{1}}_{w_{1}}\ldots\set{\Gamma_{n} \sar \Pi_{n}}_{w_{n}}$}
\DisplayProof
\begin{flushleft}
%\textbf{such that:}
such that
\end{flushleft}
\begin{itemize}

\item $\Pi_{i} := [\forb_{i,1}] A_{i,1}, \ldots, [\forb_{i,k_{i}}] A_{i,k_{i}}, B_{i,1} \imp C_{i,1}, \ldots, B_{i,n_{i}} \imp C_{i,n_{i}}, 
\Delta_{i}$;

\item $\Delta_{i} \cap \set{\wbox A, \bbox B, C \iimp D \mid A, B, C, D \in \lang} = \emptyset$;

\item $\names(\ns) = \set{w_{1}, \ldots, w_{n}}$.

\end{itemize}

\caption{Disjunctive branching rule $\dbr$.\label{fig:dbr}}
\end{figure*}

\begin{definition}[Saturated, Stable]\label{def:saturated-seq} Let $\ns$ be a nested sequent and $\tr{\ns} = (\trv, \tre, \trl)$. We define $\tr{\ns}$ to be \emph{$\cons$-saturated} \iffi it satisfies the following, for all $w \in \trv$:
\begin{description}[leftmargin=!, labelwidth=0em, labelsep=.25em] %, labelindent=0pt]

\item[$\idc$] $\infs{w,\tr{\ns}} \cap \outfs{w,\tr{\ns}} = \emptyset$;

\item[$\botlc$] $\bot \not\in \infs{w,\tr{\ns}}$;

\item[$\dislc$] if $A \lor B \in \infs{w,\tr{\ns}}$, then $\set{A,B} \cap \infs{w,\tr{\ns}} \neq \emptyset$;

\item[$\disrc$] if $A \lor B \in \outfs{w,\tr{\ns}}$, then $A, B \in \outfs{w,\tr{\ns}}$;

\item[$\conlc$] if $A \land B \in \infs{w,\tr{\ns}}$, then $A, B \in \infs{w,\tr{\ns}}$;

\item[$\conrc$] if $A \land B \in \outfs{w,\tr{\ns}}$, then $\set{A,B} \cap \outfs{w,\tr{\ns}} \neq \emptyset$; %$A \in \outfs{w,\tr{\ns}}$ or $B \in \outfs{w,\tr{\ns}}$;

\item[$\iimplc$] if $A \iimp B \in \infs{w,\tr{\ns}}$, then $A \in \outfs{w,\tr{\ns}}$ or $B \in \infs{w,\tr{\ns}}$;

\item[$\xdialc$] if $\xdia A \in \infs{w,\tr{\ns}}$, then $\exists u \in \tr{\ns}$ with $w \tre u$, $\trl(w,u) = \forb$, and $A \in \infs{u,\tr{\ns}}$;

\item[$\xdiarc$] if $\xdia A \in \outfs{w,\tr{\ns}}$ with $w \rableosx u$, then $A \in \outfs{u,\tr{\ns}}$;

\item[$\xboxlc$] if $\xbox A \in \infs{w,\tr{\ns}}$ with $w \rableosx u$, then $A \in \infs{u,\tr{\ns}}$;

\item[$\ddrc$] if $w \in \depk{\ns}{k}$ and $k \leq \mdep{\nsinp}$, then $\exists u \in \tr{\ns}$ with $w \tre u$ and $\trl(w,u) = \fd$.

\end{description}
$\tr{\ns}$ is \emph{stable} \iffi (1) $\outfs{\tr{\ns}} \cap \set{A \imp B, \wbox A, \bbox A \mid A,B \in \langt} = \emptyset$ and (2) it is saturated. Last, $\ns$ is saturated or stable \iffi $\tr{\ns}$ is saturated or stable (resp.).
\end{definition}

We now define the $\triangleleft_{w}$ operation that lets us plug data into nested sequents at certain vertices. This will be helpful in formulating our algorithm. Let the following be nested sequents: % with $X$ and $Y$ denoting their respective nestings:

\smallskip

\noindent
$\ns := \Gamma \sar \Delta, \underbrace{(\forb_{1})[\nsii_{1}], \ldots, (\forb_{n})[\nsii_{n}]}_{X}$

\smallskip

\noindent
$\nsiii := \Sigma \sar \Pi, \underbrace{(\forb_{n{+}1})[\nsiv_{n{+}1}], \ldots, (\forb_{n{+}k})[\nsiv_{n{+}k}]}_{Y}$

\noindent
We define their \emph{composition} as $\ns \odot \nsiii := \Gamma, \Sigma \sar \Delta, \Pi, X, Y$. Given a nested sequent $\ns = \ns\set{\nsii}_{w}$, we let $\ns \triangleleft_{w} \nsiii = \ns\set{\nsii \odot \nsiii}_{w}$.

Our proof-search algorithm $\prove$ is presented as Algorithm 1. Due to its length, the algorithm is split across two columns. We say that \emph{proof-search succeeds} when it outputs $\true$, and that \emph{proof-search fails} when it outputs $\false$.

Lines 1–3 implement the $\id$ and $\botl$ rules. Lines 4–6 test whether the current sequent is stable or a repeat. In either case, all further rule applications would be redundant, and the algorithm can safely halt on that branch. (NB. We will discuss repeats in detail in the next section.)

Lines 7-10, 11-13, 14-16, 17-26, 17-20, and 21-24 encode the rules $\disl$, $\disr$, $\conl$, $\conr$, $\iimpl$, respectively. Lines 25-27, 28-30, 31-33, and 34-37 encode $\xdial$, $\xdiar$, $\xboxl$, and $\ddr$, respectively. In the binary rules $\disl$, $\conr$, and $\iimpl$, the algorithm invokes two recursive calls via the line: $\prove(\ns_{1}) \ \&\& \ \prove(\ns_{2}).$ Here the symbol $\&\&$ stands for conjunction. This corresponds to \emph{conjunctive branching} in the procedure: a proof is found only if both recursive calls succeed. %return $\true$.

One novel aspect of our algorithm concerns lines 38-44, which encode simultaneous applications of the $\iimpr$ and $\xboxr$ rules. The algorithm invokes $n+k$ many recursive calls via the line: 
$$
\prove(\ns_{1}) \ \| \cdots \| \ \prove(\ns_{n+k}).
$$
Here the symbol $\|$ stands for disjunction. This corresponds to \emph{disjunctive branching} in the procedure: a proof is found if at least one recursive calls succeeds. These lines encode a bottom-up application of the $\dbr$ rule shown in Figure~\ref{fig:dbr}. Observe that $\prove$ applies rules \emph{bottom-up}; thus, when we speak of \emph{rule applications} in this section (Section~\ref{sec:decid}), 
we mean bottom-up rule applications, unless otherwise specified.

\begin{example} To demonstrate the functionality of $\dbr$, we give an example application with principal formulae $B \imp C$, $\wbox F$, and $\bbox G$:
\begin{center}
%\resizebox{\textwidth}{!}{
\AxiomC{$A, B \sar C, (\fd)[D, E \sar \ ]
\quad
A \sar (\fd)[D, E \sar (\fd)[\ \sar F]]
\quad
\ns
$}
\RightLabel{$\dbr$}
\UnaryInfC{$A \sar B \imp C, (\fd)[D, E \sar \wbox F, \bbox G]$}
\DisplayProof
\end{center}
where $\ns = A \sar (\fd)[D, E \sar (\bd)[\ \sar G]]$.
\end{example}

Since lines 38–44 correspond to applications of the $\dbr$ rule, we may regard the structure generated by a run of $\prove$ as a derivation that employs all rules in $(\calc \setminus \set{\iimpr,\xboxr}) \cup \set{\dbr}$. This perspective motivates the definition of a \emph{computation tree}, a structure that plays a crucial role in extracting both proofs and counter-models from terminating proof-search.

\begin{definition}[Computation Tree]\label{def:computation-tree} A \emph{computation tree} (relative to $\calc$) is a structure $\ct := (\ctv,\cte,\ctl)$ such that $\ctv$ is a non-empty set of seq-trees, $\cte \ \subseteq \ctv \times \ctv$, and $\ctl : \ctv \to \set{\tval,\fval}$, which satisfies the following condition: each parent node $\ltr \in \ctv$ is the conclusion of a rule in $(\calc \setminus \set{\iimpr,\xboxr}) \cup \set{\dbr}$ with its children (if they exist) the corresponding premises. %We let $\exrel^{+}$ and $\exrel^{*}$ be the transitive and reflexive-transitive closure of $\exrel$, respectively.
\end{definition}

%We observe that a computation tree corresponds to a derivation in the `calculus' $(\calc \setminus \set{\iimpr,\wboxr,\bboxr}) \cup \set{\dbr}$. 

To establish the correctness of our proof-search algorithm (Theorems \ref{thm:success-gives-proof} and~\ref{thm:false-implies-not-prove}), it is useful to have access to the \emph{entire} computation tree generated by $\prove(\nsinp)$, i.e., the full tree structure consisting of all branches explored during a terminating run of the algorithm. We therefore define $\ct(\nsinp) := (\ctv,\cte,\ctl)$ to be the computation tree built by $\prove(\nsinp)$ during its execution such that for all $\tr{\ns} \in \ctv$, (1) $\ctl(\tr{\ns}) = \tval$ \iffi $\prove(\ns) = \true$ and (2) $\ctl(\tr{\ns}) = \fval$ \iffi $\prove(\ns) = \false$. For completeness, the formal definition of $\ct(\nsinp)$ is given in the appendix.

%As shown in the figure, one can process $\ct(\nsinp)$ in a root-first manner to prune extraneous branches in $\ct(\nsinp)$. This is accomplished by selecting and retaining a single premise of a $\dbr$ application. In the extracted proof, shown right, observe that pruning $\dbr$ applications yields $\iimpr$, $\wboxr$, or $\bboxr$ applications. We now prove that this pruning operation can be done in general to extract a proof of the input when proof-search succeeds.

\begin{figure*}
    \centering

\begin{center}
\begin{tabular}{c}
\AxiomC{$q \sar r : \fval$}

\AxiomC{}
\RightLabel{$\botl$}
\UnaryInfC{$\sar (\fd)[ \sar (\fd) [ \bot \sar \bot]] : \tval$}
\RightLabel{$\dbr$}
\UnaryInfC{$\sar (\fd)[ \sar (\fd) [ \sar \bot \iimp \bot]] : \tval$}

\AxiomC{$\sar (\fd)[ \sar (\bd)[ \sar \bot] : \fval$}
\RightLabel{$\dbr$}
\BinaryInfC{$\sar (\fd)[ \sar \wbox (\bot \iimp \bot), \bbox \bot] : \tval$}
\RightLabel{$\disr$}
\UnaryInfC{$\sar (\fd)[ \sar \wbox (\bot \iimp \bot) \lor \bbox \bot] : \tval$}

\RightLabel{$\dbr$}
\BinaryInfC{$\sar q \iimp r, \wbox (\wbox (\bot \iimp \bot) \lor \bbox \bot) : \tval$}
\RightLabel{$\disr$}
\UnaryInfC{$\sar (q \iimp r) \lor \wbox (\wbox (\bot \iimp \bot) \lor \bbox \bot) : \tval$}
\DisplayProof
\end{tabular}
\end{center}
\begin{center}
\begin{tabular}{c}
\AxiomC{}
\RightLabel{$\botl$}
\UnaryInfC{$\sar (\fd)[ \sar (\fd) [ \bot \sar \bot]]$}
\RightLabel{$\iimpr$}
\UnaryInfC{$\sar (\fd)[ \sar (\fd) [ \sar \bot \iimp \bot]]$}
\RightLabel{$\wboxr$}
\UnaryInfC{$\sar (\fd)[ \sar \wbox (\bot \iimp \bot), \bbox \bot]$}
\RightLabel{$\disr$}
\UnaryInfC{$\sar (\fd)[ \sar \wbox (\bot \iimp \bot) \lor \bbox \bot]$}

\RightLabel{$\wboxr$}
\UnaryInfC{$\sar q \iimp r, \wbox (\wbox (\bot \iimp \bot) \lor \bbox \bot)$}
\RightLabel{$\disr$}
\UnaryInfC{$\sar (q \iimp r) \lor \wbox (\wbox (\bot \iimp \bot) \lor \bbox \bot)$}
\DisplayProof
\end{tabular}
\end{center}

\caption{Computation tree and proof extraction example.\label{fig:ct-example-and-proof-extract}}
\end{figure*}

\begin{example} Let $\nsinp := \ \sar (q \iimp r) \lor \wbox (\wbox (\bot \iimp \bot) \lor \bbox \bot)$. The computation tree $\ct(\nsinp) = (\ctv,\cte,\ctl)$ corresponding to $\prove(\nsinp)$ is shown top in \fig~\ref{fig:ct-example-and-proof-extract}. We have included the labels $\tval$ and $\fval$ to the right of sequents, arising from the labeling function $\ctl$. 

As shown in the figure, one can process $\ct(\nsinp)$ in a root-first manner to prune extraneous branches in $\ct(\nsinp)$. This is accomplished by selecting and retaining a single premise of a $\dbr$ application. In the extracted proof, shown bottom, observe that pruning $\dbr$ applications yields $\iimpr$ or $\xboxr$ instances. We now prove that this pruning operation can be done in general to extract a proof of the input when proof-search succeeds.
\end{example}

\begin{theorem}\label{thm:success-gives-proof}
If $\prove(\nsinp) = \true$, then $\calc \proves \nsinp$.
\end{theorem}

\begin{proof} We suppose that $\prove(\nsinp) = \true$ and let $\ct(\nsinp) := (\ctv,\cte,\ctl)$ be the computation tree built by $\prove(\nsinp)$. Let us construct a proof $\prf := (\ctv',\cte')$ with $\ctv' \subseteq \ctv$ and $\cte' \subseteq \cte$ by processing $\ct(\nsinp)$ in a root-first manner and pruning extraneous branches. We define $\prf = (\ctv',\cte')$ as follows:
\begin{description}[leftmargin=!, labelwidth=0em, labelsep=.25em]

\item[$(1)$] Let $\nsinp \in \ctv'$ and observe that $\ctl(\nsinp) = \tval$ by assumption;

\item[$(2)$] If $\ctl(\ns) = \tval$ and $\ns \in \ctv'$ is the conclusion of a unary rule $ \ru \in \calc \setminus \set{\iimpr,\xboxr}$ in $\ct(\nsinp)$ with $\ns' \in \ctv$ the premise, then let $\ns' \in \ctv'$ and $(\ns,\ns') \in \cte'$;

\item[$(3)$] If $\ctl(\ns) = \tval$ and $\ns \in \ctv'$ is the conclusion of a binary rule $\ru \in \calc \setminus \set{\iimpr,\xboxr}$ in $\ct(\nsinp)$ with $\ns', \ns'' \in \ctv$ the premises, then let $\ns', \ns'' \in \ctv'$ and $(\ns,\ns'), (\ns,\ns'') \in \cte'$;

\item[$(4)$] If $\ctl(\ns) = \tval$ and $\ns \in \ctv'$ is the conclusion of $\dbr$ in $\ct(\nsinp)$ with $\ns_{1}, \ldots, \ns_{n} \in \ctv$ the premises, then we choose a premise $\ns_{i}$ such that $\ctl(\ns_{i}) = \tval$ (which is guaranteed to exist because $\ctl(\ns) = \tval$), and let $\ns_{i} \in \ctv'$ and $(\ns,\ns_{i}) \in \cte'$. 

\end{description}
In the above definition, one starts at the root of $\ct(\nsinp)$ and retains rule applications in $\calc \setminus \set{\iimpr,\xboxr}$ via clauses (2) and (3), while pruning branches in clause (4) and only retaining a single premise so that the $\dbr$ application becomes either an $\iimpr$ or $\xboxr$ application. Thus, all rule applications in $\prf$ will be in $\calc$. Moreover, every nested sequent in $\prf = (\ctv',\cte')$ is guaranteed to be labeled with $\tval$ via $\ctl$ by definition; hence, all leaves are instances of $\id$ or $\botl$.
\end{proof}

\subsection{Detecting Repeats and Termination}\label{subsec:repeats-termination}

In this section, we introduce a new loop-checking mechanism that ensures termination of proof-search. We first prove that the depth of all nested sequents generated during proof-search is uniformly bounded. Despite this bound, nested sequents may still reoccur %(up to homomorphic equivalence) 
along a branch of a computation tree generated by the algorithm. Detecting and handling such repetitions--henceforth called \emph{repeats}--is therefore essential for guaranteeing termination.

A repeat arises along a branch of a computation tree when a nested sequent is homomorphically equivalent to one of its ancestors on that branch. This differs from standard loop-checking techniques in nested sequent calculi (cf.~\cite{Bru09,TiuIanGor12}), which operate \emph{within} a nested sequent and identify repeating nodes, rather than distinct nested sequents along branches.

\begin{lemma}\label{lem:dpeth-bounded}
Let $\ct(\nsinp) = (\ctv, \cte, \ctl)$ be the computation tree of $\prove(\nsinp)$. Then, the following two statements hold:
\begin{description}[leftmargin=!, labelwidth=0em, labelsep=.25em]

\item[$(1)$] If $\ns \in \ctv$ and $B \in \depk{\ns}{k}$, then $\mdep{B} \leq \mdep{\nsinp} - k$;

%\item[(2)] If $\D \in \axs$, $\ns \in \exv$ is saturated, and $w$ is a leaf in $\ns$, then $w \in \depk{\ns}{k+1}$ with $\mdep{\nsinp} = k$.

\item[$(2)$] If $\D \in \cons$ and $\ns \in \ctv$ is saturated, then $w$ is a leaf in $\ns$ \iffi $w \in \depk{\ns}{k+1}$ with $\mdep{\nsinp} = k$.

\end{description}
\end{lemma}

\begin{proof} Let us prove (1) first. We proceed by a root-first induction on $\ct(\nsinp)$. % that if $\nsii \in \exv$ and $B \in \depk{\nsii}{k}$ then $\mdep{B} \leq \mdep{\ns} - k$.

\textit{Base case.} Let $B \in \depk{\nsinp}{k}$. Then, $\mdep{B} + k \leq \mdep{\nsinp}$ by definition, from which the result immediately follows.

\textit{Inductive step.} We make a case distinction based on the last rule applied in $\ct(\ns)$ and only consider the $\xdiar$ and $\dbr$ cases as the remaining cases are similar.

$\xdiar$. Observe that $\xdiar$ has one of the five forms shown below, depending on the content of $\cons$. The rules $\xdiari$ and $\xdiarii$ are valid instances of $\xdiar$ for all $\cons \subseteq \set{\T,\B,\D}$, $\xdiar_{3}$ and $\xdiar_{4}$ are valid instances when $\B \in \cons$, and $\xdiar_{5}$ is a valid instance when $\T \in \cons$.
\begin{center}
\begin{tabular}{c}
\AxiomC{$\ns \sbl \Gamma \sar \xdia A, \Delta, (\forb)[\Sigma, A \sar \Pi] \sbr$}
\RightLabel{$\xdiari$}
\UnaryInfC{$\ns \sbl \Gamma \sar \xdia A, \Delta, (\forb)[\Sigma \sar \Pi] \sbr$}
\DisplayProof
\end{tabular}
\end{center}
\begin{center}
\AxiomC{$\ns \sbl \Gamma \sar A, \Delta, (\conv{\forb})[\Sigma, \xdia A \sar \Pi] \sbr$}
\RightLabel{$\xdiarii$}
\UnaryInfC{$\ns \sbl \Gamma \sar \Delta, (\conv{\forb})[\Sigma, \xdia A \sar \Pi] \sbr$}
\DisplayProof
\end{center}
\begin{center}
\begin{tabular}{c}
\AxiomC{$\ns \sbl \Gamma \sar \xdia A, \Delta, (\conv{\forb})[\Sigma, A \sar \Pi] \sbr$}
\RightLabel{$\xdiar_{3}$}
\UnaryInfC{$\ns \sbl \Gamma \sar \xdia A, \Delta, (\conv{\forb})[\Sigma \sar \Pi] \sbr$}
\DisplayProof
\end{tabular}
\end{center}
\begin{center}
\begin{tabular}{@{\hskip .2em} c @{\hskip .2em} c}
\AxiomC{$\ns \sbl \Gamma \sar A, \Delta, (\forb)[\Sigma \sar \xdia A, \Pi] \sbr$}
\RightLabel{$\xdiar_{4}$}
\UnaryInfC{$\ns \sbl \Gamma \sar \Delta, (\forb)[\Sigma \sar \xdia A, \Pi] \sbr$}
\DisplayProof

&

\AxiomC{$\ns \sbl \Gamma \sar \xdia A, A, \Delta \sbr$}
\RightLabel{$\xdiar_{5}$}
\UnaryInfC{$\ns \sbl \Gamma \sar \xdia A, \Delta \sbr$}
\DisplayProof
\end{tabular}
\end{center}
Suppose that IH holds for the conclusion of one of these rules. We argue that the claim holds for the premise. In each case, if $B$ is not the auxiliary formula $A$ in the premise of the rule, then the claim immediately holds. Suppose then that $B$ is the auxiliary formula $A$. In the $\xdiar_{1}$ and $\xdiar_{3}$ cases, we have that $A \in \depk{\nsii}{k+1}$ in the premise $\nsii$. The claim follows as shown below, where the equality holds due to the definition of the modal depth and the inequality follows from IH.
$$
\mdep{A} = \mdep{\xdia A} - 1 \leq \mdep{\nsii} - (k + 1)
$$
Next, because $\mdep{A} \leq \mdep{\xdia A}$, $A \in \depk{\nsii}{k-1}$ in $\xdiar_{2}$ and $\xdiar_{4}$, and $A \in \depk{\nsii}{k}$ in $\xdiar_{5}$, the claim immediately follows by IH in all of these cases.

%$\mdep{A} \leq \mdep{\xdia A} \leq \mdep{\nsii} - k \leq \mdep{\nsii} - k + 1 \leq \mdep{\nsii} - (k - 1)$
%$\mdep{A} \leq \mdep{\xdia A} \leq \mdep{\nsii} - k$

$\dbr$. To not over-complicate matters, we only consider a specific instance of the $\dbr$ rule that has two principal formulae $\xbox A$ and $C \iimp D$. The general case is argued similarly. Suppose $\dbr$ is of the form shown below and that IH holds for the conclusion:
\begin{center}
%\resizebox{\textwidth}{!}{
\AxiomC{$\nsii^{\da}\set{\Gamma \sar (\charx)[ \ \sar A]}$}
\AxiomC{$\nsii^{\da}\set{\Gamma, C \sar D}$}
\RightLabel{$\boximpr$}
\BinaryInfC{$\nsii\set{\Gamma \sar \xbox A, C \iimp D, \Delta}$}
\DisplayProof
%}
\end{center}
If $B$ is not auxiliary in either premise of the rule, then the claim immediately follows. If $B$ is the auxiliary formula $C$ or $D$ in the right premise, then the claim follows from IH since $\mdep{C} = \mdep{D} = \mdep{C \imp D} \leq \mdep{\nsii} - k$. If $B$ is the auxiliary formula $A$ in the left premise, then observe that $A \in \depk{\nsii}{k+1}$. The claim follows by the definition of modal depth and IH:
$$
\mdep{A} = \mdep{\xbox A} - 1 \leq \mdep{\nsii} - (k + 1).
$$

(2) Since $\D \in \axs$ and $\ns \in \ctv$ is saturated, we know that $\ns$ satisfies condition $\ddrc$, and so, if $w$ is a leaf in $\ns$, then $w \in \depk{\ns}{n}$ with $n \geq k + 1$ such that $\mdep{\nsinp} = k$. Let $n = m + k + 1$ with $m \in \mathbb{N}$. Observe that the only rules that increase the depth of a nested sequent (when applied bottom-up) are $\xdial$, $\xboxr$, and $\ddr$. 

By part (1) of the lemma, we know that if $B \in \depk{\ns}{n}$, then
$$
\mdep{B} \leq \mdep{\nsinp} - n = k - (m + k + 1) = - (m+1).
$$
However, it is not possible for the modal depth of a formula to a negative number. Therefore, we know that for any $u \in \tr{\ns}$ occurring at a depth of $k+1$ that $\Gamma = \Delta = \emptyset$ for $\trl(u) = (\Gamma \sar \Delta)$. In other words, all components at a depth greater than $k$ must be empty Gentzen sequents. It follows that it is not possible to increase the depth of a nested sequent beyond $k$ during proof-search by applying either a $\xdial$ or $\xboxr$ rule since such rules require principal formulae. By our observation then, the $w$-component must have been introduced by a bottom-up application of $\ddr$. Yet, $\prove$ only bottom-up applies $\ddr$ (lines 39-42) to nodes that are at a depth of at most $k$, meaning, $w$ cannot occur at a depth greater than $k+1$ due to a $\ddr$ application. It must be the case that $n = k+1$.
\end{proof}

To support both repeat detection and the extraction of counter-models from failed proof-search, we now introduce two kinds of homomorphisms that will play a central role in our work.

\begin{definition}[Morphism] Let $\ltr = (\trv,\tre,\trl)$ and $\ltr' = (\trv',\tre',\trl')$ be seq-trees with $w$ and $w'$ their respective roots. We define a function $\wkhom : \trv \to \trv'$ %from $\ltr$ to $\ltr'$ 
to be a \emph{weak morphism} \iffi 
\begin{description}

\item[$(1)$] $\wkhom(w) = w'$;

\item[$(2)$] if $u \tre v$, then $\wkhom(u) \tre' \wkhom(v)$; 

\item[$(3)$] $\trl(u,v) = \trl(u',v')$.

\end{description}
We define a function $\sthom : \trv \to \trv'$ to be a \emph{strong morphism} \iffi (1) it is a weak morphism and (2) $\trl(u) = \trl'(\sthom(u))$ for all $u \in \trv$. 

We define a \emph{morphism} $\mor$ to be either a weak or strong morphism. % and we often write $\mor : \ns \to \nsii$ (rather than $\mor : T \to T'$) to indicate that $\mor$ is a function from $\ns$ to $\nsii$.
Given a sequence $\sigma := \mor_{1}, \ldots, \mor_{n}$ of morphisms such that the composition $\mor_{i+1} \circ \mor_{i}$ is defined and is a morphism for $i \in [n{-}1]$, we define $\mor_{\sigma} := \mor_{n} \circ \cdots \circ \mor_{1}$.

$\ltr$ and $\ltr'$ are \emph{morphically equivalent}, written $\ltr \hequiv \ltr'$, \iffi there exist strong morphisms $\wkhom : \trv \to \trv'$ and $\wkhom': \trv' \to \trv$. %The equivalence class of nested sequents homomorphically equivalent to $\ns$ is defined accordingly: 
Last, we define the equivalence class $\hec{\ltr} := \set{\ltr' \mid \ltr \hequiv \ltr' \text{ and $\ltr'$ is a seq-tree.}}$.
\end{definition}

We define $S(\nsinp)$ to be the set of all seq-trees $\ltr = (\trv,\tre,\trl)$ such that for all $w \in \trv$ with $\trl(w) = \Gamma \sar \Delta$, $\Gamma, \Delta \subseteq \sufo(\nsinp)$. $S_{k}(\nsinp)$ denotes the set of all seq-trees in $S(\nsinp)$ whose depth is at most $k$.

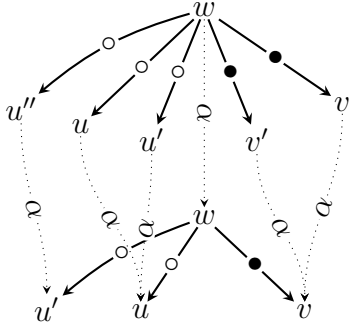
\begin{figure}
\centering

%\begin{center}
\begin{tikzpicture}[
    >=stealth,
    node distance=14mm,
    every node/.style={inner sep=1pt},
    lab/.style={midway, sloped, fill=white, inner sep=1pt}
]

% --- FIRST DIAGRAM ------------------------------------------------
\node (w) {$w$};
\node (u) [below left=of w, xshift=5mm] {$u$};
\node (v) [below right=of w] {$v$};
\node (u') [below left=of w, xshift=-7mm, yshift=1mm] {$u'$};

\draw[->,thick] (w) to[bend left=0] node[lab] {$\fd$} (u);
\draw[->,thick] (w) -- node[lab] {$\bd$} (v);
\draw[->,thick] (w) to[bend right=10] node[lab] {$\fd$} (u');

% --- SECOND DIAGRAM (PRIMED) --------------------------------------
\node (w1) [above=25mm of w] {$w$};

\node (u1) [below left=of w1, xshift=-3mm, yshift=-3mm] {$u$};
\node (v1) [below right=of w1,, xshift=5mm] {$v$};

\node (u2) [below=of w1, xshift=-7mm] {$u'$};
\node (v2) [below=of w1, xshift=7mm] {$v'$};

\node (u3) [below left=of w1, xshift=-10mm] {$u''$};

\draw[->,thick] (w1) to[bend left=0] node[lab] {$\fd$} (u1);
\draw[->,thick] (w1) -- node[lab] {$\bd$} (v1);
\draw[->,thick] (w1) -- node[lab] {$\fd$} (u2);
\draw[->,thick] (w1) -- node[lab] {$\bd$} (v2);
\draw[->,thick] (w1) to[bend right=12] node[lab] {$\fd$} (u3);

% --- ST-HOMOMORPHISM ARROWS --------------------------------------
\draw[->, dotted] (w1) to[out=270, in=90, looseness=1] node[lab] {$\sthom$} (w);

\foreach \a/\b in {u1/u, u2/u, u3/u', v1/v, v2/v}
    \draw[->, dotted] (\a) to[out=270, in=90, looseness=1] node[lab] {$\sthom$} (\b);

\end{tikzpicture}
%\end{center}

\caption{Repeat example.\label{fig:repeat-example}}
\end{figure}

\begin{lemma}\label{lem:morpic-properties} Let $\ltr_{i} = (\trv_{i},\trep{i},\trl_{i})$ be seq-tress in $\ct(\nsinp)$ for $i \in [3]$. 
\begin{description}[leftmargin=!, labelwidth=0em, labelsep=.25em]

\item[$(1)$] The relation $\hequiv$ is an equivalence relation;

\item[$(2)$] $S_{k}(\nsinp) / \hequiv$ is a finite set, that is, there are only finitely many equivalence classes of seq-trees from $S_{k}(\nsinp)$ of depth $k$; %equivalence classes $\hec{\ltr}$ of seq-trees of depth $k$ in $S_{k}(\nsinp) / \hequiv$; %There exist at most $\mathcal{O}(2^{k})$ equivalence classes of nested sequents of depth $k$.

\item[$(3)$] if $\wkhom : \trv_{1} \to \trv_{2}$ and $\wkhom' : \trv_{2} \to \trv_{3}$ are weak morphisms, then $\wkhom' \circ \wkhom$ is a weak morphism;

\item[$(4)$] if $\wkhom : \trv_{1} \to \trv_{2}$ is a weak morphism and $\ltr_{1} \models w \rable^{\cons}_{\forb} u$, then it follows that $\ltr_{2} \models \wkhom(w) \rable^{\cons}_{\forb} \wkhom(u)$.

\item[$(5)$] if $\D \in \cons$, $\wkhom : \trv_{1} \to \trv_{2}$ is a weak morphism, $\ltr_{1}$ and $\ltr_{2}$ are saturated, and $w$ is a leaf in $\ltr_{1}$, then $\wkhom(w)$ is a leaf in $\ltr_{2}$.

\end{description}
\end{lemma}

\begin{proof} Claims (1)--(3) are straightforward, so we argue claims (4) and (5).

(4) Let $\wkhom : \trv_{1} \to \trv_{2}$ be a weak morphism. Assume $\ltr_{1} \models w \rable_{\forb}^{\cons} u$. It follows that $(w,u) \in \ltr_{1}^{\cons}$ with the label $\forb$. Regardless of if $\B \in \cons$ or $\T \in \cons$, we know that $(\wkhom(w),\wkhom(u)) \in \ltr_{2}^{\cons}$ with the label $\forb$ because $\wkhom$ is a weak morphism from $\ltr_{1}$ to $\ltr_{2}$. Therefore, $\ltr_{2} \models \wkhom(w) \rable^{\cons}_{\forb} \wkhom(u)$.

(5) Let $\mdep{\nsinp} = k$. By Lemma~\ref{lem:dpeth-bounded}, we know that if $w$ is a leaf of $\ltr_{1}$, then $w \in \depk{\ltr_{1}}{k+1}$. By the definition of a weak morphism, it follows that $\wkhom(w) \in \depk{\ltr_{2}}{k+1}$. By Lemma~\ref{lem:dpeth-bounded}, it follows that $\wkhom(w)$ is a leaf of $\ltr_{2}$.
\end{proof}

When extracting counter-models from failed proof-search, the set of worlds will be determined by the names occurring in distinct $\cons$-saturated sequents. Since different nested sequents may reuse the same names, we introduce a systematic renaming mechanism to distinguish these worlds in the extracted counter-model. Renaming ensures that worlds originating from different sequents are not inadvertently identified, thereby ensuring the correctness of our counter-model construction.

\begin{definition}[Enumeration, Renaming] We define an \emph{enumeration} $\enum = \ltr_{1}, \ldots, \ltr_{n}$ to be a finite list of seq-trees. For a seq-tree $\ltr_{i} \in \enum$, we define the \emph{$\enum$-renaming} of $\ltr_{i} = (\trv,\tre,\trl)$ to be the labeled tree $\enum\ltr_{i} = (\trv',\tre',\trl')$ such that 
\begin{description}[leftmargin=!, labelwidth=0em, labelsep=.25em]

\item[$(1)$] $\trv' := \trv \times \set{i}$;

\item[$(2)$] $(w,i) \tre' (u,i)$ \iffi $w \tre u$;

\item[$(3)$] $\trl(w,i) := \trl(w)$ and $\trl((w,i),(u,i)) := \trl(w,u)$.

\end{description}
\end{definition}

%Observe that all an $\enum$-renaming does is associate the position $i$ of $\ltr_{i}$ in $\enum$ with each vertex. 

If one inspects the rules in $\calc$, they will observe that components are never deleted in bottom-up rule applications. That is to say, if $\ns$ occurs below $\nsii$ in a branch of a derivation, then $\trv \subseteq \trv'$ and $\tre \subseteq \tre'$ for $\tr{\ns} = (\trv,\tre,\trl)$ and $\tr{\nsii} = (\trv',\tre',\trl')$. So, even though the labels of vertices may change, the vertices themselves and their edge structure is always preserved bottom-up. This induces a canonical weak morphism that we call the \emph{natural morphism}, which plays a central role in extracting counter-models from failed proof-search.

\begin{definition}[Natural Morphism] Let $\enum$ be an enumeration of seq-trees containing $\ltr_{i} = (\trv_{i},\trep{i},\trl_{i})$ and $\ltr_{j} = (\trv_{j},\trep{j},\trl_{j})$ at positions $i$ and $j$, respectively, such that $\trv_{i} \subseteq \trv_{j}$, $\trep{i} \subseteq \trep{j}$, and for each $w \trep{i} u$, $\trl_{i}(w,u) = \trl_{j}(w,u)$. We define the \emph{natural morphism} $\homnat$ from $\enum\ltr$ to $\enum\ltr'$ to be the weak morphism such that $\homnat(w,i) = (w,j)$.
\end{definition}

\begin{remark} A natural morphism always exists from a $\enum$-renamed seq-tree to itself, namely, the identity function.
\end{remark}

\begin{definition}[Repeat]\label{def:repeat} Let $\ct = (\ctv,\cte,\ctl)$ be a computation tree. We define $\ltr \in \ctv$ to be a \emph{repeat} \iffi (1) it is a leaf, (2) it is saturated, and (3) there exists a $\ltr'$ (called a \emph{companion} of the repeat) such that $\ltr'$ is an ancestor %(w.r.t $\cte$) 
of $\ltr$ and a strong morphism $\sthom$ exists from $\ltr$ to $\ltr'$.
%and $\ns \in \hec{\nsii}$. 
\end{definition}

\begin{example} Let us consider $\prove(\neg \wbox p, \neg \bbox q \sar  \ )$. This builds a derivation $\prf$ of the input that has a branch (shown below) with a repeat $\ns$ and whose companion we denote by $\nsii$.
\begin{center}
%\resizebox{\columnwidth}{!}{
\AxiomC{$\neg \wbox p, \neg \bbox q \sar \wbox p, \bbox q, (\fd)[\  \sar p], (\bd)[ \  \sar \ ], (\fd)[ \  \sar  \ ], (\bd)[ \  \sar  \ ], (\fd)[ \  \sar  \ ] \ \ \mathbf{(= \ns)}$}
\RightLabel{$\iimpl \times 2$}
\UnaryInfC{$\neg \wbox p, \neg \bbox q \sar (\fd)[\  \sar p], (\bd)[ \  \sar \ ], (\fd)[ \  \sar  \ ], (\bd)[ \  \sar  \ ], (\fd)[ \  \sar  \ ]$}
\RightLabel{$\wboxr$}
\UnaryInfC{$\neg \wbox p, \neg \bbox q \sar \wbox p, \bbox q,(\bd)[ \  \sar q], (\fd)[ \  \sar  \ ], (\bd)[ \  \sar  \ ], (\fd)[ \  \sar  \ ]$}
\RightLabel{$\iimpl \times 2$}
\UnaryInfC{$\neg \wbox p, \neg \bbox q \sar (\bd)[ \  \sar q], (\fd)[ \  \sar  \ ], (\bd)[ \  \sar  \ ], (\fd)[ \  \sar  \ ]$}

\RightLabel{$\bboxr$}
\UnaryInfC{$\neg \wbox p, \neg \bbox q \sar \wbox p, \bbox q, (\fd)[ \  \sar p], (\bd)[ \  \sar  \ ], (\fd)[ \  \sar  \ ]  \ \ \mathbf{(= \nsii)}$}
\RightLabel{$\iimpl \times 2$}
\UnaryInfC{$\neg \wbox p, \neg \bbox q \sar (\fd)[ \  \sar p], (\bd)[ \  \sar  \ ], (\fd)[ \  \sar  \ ]$}

\RightLabel{$\wboxr$}
\UnaryInfC{$\neg \wbox p, \neg \bbox q \sar \wbox p, \bbox q, (\bd)[ \  \sar q], (\fd)[ \  \sar  \ ]$}
\RightLabel{$\iimpl \times 2$}
\UnaryInfC{$\neg \wbox p, \neg \bbox q \sar (\bd)[ \  \sar q], (\fd)[ \  \sar  \ ]$}

\RightLabel{$\bboxr$}
\UnaryInfC{$\neg \wbox p, \neg \bbox q \sar \wbox p, \bbox q, (\fd)[ \  \sar p]$}
\RightLabel{$\iimpl \times 2$}
\UnaryInfC{$\neg \wbox p, \neg \bbox q \sar (\fd)[ \  \sar p]$}

\RightLabel{$\wboxr$}
\UnaryInfC{$\neg \wbox p, \neg \bbox q \sar \wbox p, \bbox q$}
\RightLabel{$\iimpl \times 2$}
\UnaryInfC{$\neg \wbox p, \neg \bbox q \sar  \ $}
\DisplayProof
%}
\end{center}
Let $\Gamma := \neg \wbox p, \neg \bbox q$ and $\Delta := \wbox p, \bbox q$. The names of the components of $\nsii$ are listed below:
%Its companion $\nsii$ will be as follows, where the names of components are listed underneath:

\smallskip

\noindent
$\underbrace{\Gamma \sar \Delta}_{w}, (\fd)[\underbrace{  \sar p}_{u'}], (\bd)[\underbrace{  \sar }_{v}], (\fd)[\underbrace{  \sar  }_{u}]$

\smallskip

\noindent
The names of the components of the repeat $\ns$ are listed below:
%The root of the repeat $\ns$ is named $w$ and the remaining components in the consequent have the following names:

\smallskip

\noindent
$\underbrace{\Gamma \sar \Delta}_{w}, (\fd)[\underbrace{  \sar  p}_{u''}], (\bd)[\underbrace{\sar}_{v'}], (\fd)[\underbrace{ \  \sar  \ }_{u'}], (\bd)[\underbrace{ \  \sar  \ }_{v}], (\fd)[\underbrace{ \  \sar  \ }_{u}]$

\smallskip

\noindent
Figure \ref{fig:repeat-example} shows graphical representations of $\tr{\ns}$ (top) and $\tr{\nsii}$ (bottom). Vertex labels are omitted for readability. The dotted arrows indicate a function $\sthom$, which one can readily verify is a strong morphism from $\tr{\ns}$ to $\tr{\nsii}$.

This example highlights a central way in which non-termination can arise during proof-search. Although $\ns \hequiv \nsii$, the nested sequent $\ns$ has strictly larger outdegree than $\nsii$. Without repeat detection, the branch of $\prf$ containing $\ns$ and $\nsii$ would continue to grow indefinitely, generating infinitely many nested sequents in the equivalence class $\hec{\ns}$ ($ = \hec{\nsii}$) with unbounded outdegree.

The root of this phenomenon lies in the behavior of the non-invertible $\xboxr$ rule, which deletes consequents when applied bottom-up. As a result, information may be lost during bottom-up rule applications, enabling the construction of nested sequents with increasingly large outdegrees. This kind of behavior does not occur in the classical setting (cf. \cite{Bru09,TiuIanGor12}).
%This implies that formulae of the form $\wbox A$ and $\bbox B$, which previously 
\end{example}

\begin{theorem}\label{thm:termination}
$\prove(\nsinp)$ terminates. % \tim{in $\twoexptime$.}
\end{theorem}

\begin{proof} Let $\ct(\nsinp) = (\ctv,\cte,\ctl)$ be the computation tree for $\prove(\nsinp)$. Since the rules used in $\prove$ are analytic, we know that all seq-trees in $\ctv$ will occur in $S(\nsinp)$. Hence, there are only finitely many Gentzen sequents that can serve as labels of vertices in a seq-tree of $\ct(\nsinp)$. By this fact and Lemma~\ref{lem:dpeth-bounded}, every nested sequent that can be generated during proof-search has its depth bounded by a finite number $k$. By Lemma~\ref{lem:morpic-properties}-(2), there are only finitely many equivalence classes $\hec{\ltr_{1}}, \ldots, \hec{\ltr_{n}}$ in $S_{k}(\nsinp)$. By the pigeonhole principle, if branches in $\ct(\nsinp)$ were to continue to grow, eventually two members of the same equivalence would occur along the same branch. It follows that the height of $\ct(\nsinp)$ is bounded, meaning, $\prove(\nsinp)$ must eventually terminate. %\tim{Add complexity analysis}
\end{proof}

\subsection{Counter-Model Extraction}\label{subsec:counter-model-extract}

We now explain how to extract counter-models from computation trees when proof-search fails. \emph{For the remainder of this subsection, we assume $\prove(\nsinp) = \false$ and fix $\ct(\nsinp) = (\ctv,\cte,\ctl)$ to be the corresponding computation tree.} We define yet another structure, called a \emph{blueprint}, from which a counter-model for $\nsinp$ may be finally extracted. Let $\enum$ be an enumeration of the seq-trees in $\ctv$. We define $\bp(\nsinp) := (\bpv,\bps,\bpe)$ to be the \emph{blueprint} of $\prove(\nsinp)$ such that
\begin{description}

\item[$(1)$] $\bpv$ is the smallest set satisfying:

\end{description}
\begin{description}

\item[$\quad (a)$] $\enum\tr{\nsinp} \in \bpv$;

\item[$\quad (b)$] if $\enum\tr{\ns} \in \bpv$, $\tr{\ns} \cte \tr{\nsii}$ and $\ctl(\tr{\nsii}) = \fval$, then $\enum\tr{\nsii} \in \bpv$.
    
\end{description}
\begin{description}

\item[$(2)$] $\bps := \set{\enum\tr{\ns} \in \bpv \mid \tr{\ns} \text{ is saturated.}}$

\item[$(3)$] $\enum\tr{\ns} \bpe \enum\tr{\nsii}$ \iffi $\tr{\ns} \cte \tr{\nsii}$.

\end{description}
\emph{For the remainder of this subsection, we fix $\bp(\nsinp)$ to be the blueprint of $\prove(\nsinp)$.} By conditions (1) and (3), one can see that $\bp(\nsinp)$ is obtained by starting at the root of $\ct(\nsinp)$ and by taking the downward closure of seq-trees labeled with $\fval$.

Let $\bptr{\ns_{1}}, \ldots, \bptr{\ns_{n}}$ be all repeats in $\bp(\nsinp)$, and let $\vechom = \sthom_{1}, \ldots, \sthom_{n}$ be a list of strong morphisms such that each $\sthom_{i}$ maps the repeat $\bptr{\ns_{i}}$ to its corresponding companion.\footnote{Although more than one strong morphism may exist between a repeat and its companion, selecting any one of them suffices for the counter-model construction.} \emph{For the remainder of this subsection, we fix $\vechom$ as the chosen family of strong morphisms mapping repeats to their companions in $\bp(\nsinp)$.}

We now introduce certain kinds of paths in $\bp(\nsinp)$ that will be used to define the intuitionistic relation $\leq$ in our counter-model. Specifically, we distinguish between two types of edges that may occur along such paths: edges induced by natural morphisms and edges induced by the strong morphisms in $\vechom$:

\smallskip

\noindent
(1) $(\pathi,i) \hreach{\homnat} (\pathii,j)$ \iffi there exist $\bptr{\ns}, \bptr{\nsii} \in \bps$ s.t. $\pathi = \pathii$, $\homnat$ is the natural morphism from $\bptr{\ns}$ to $\bptr{\nsii}$, and $\bptr{\ns} \bpe^{*} \bptr{\nsii}$;\footnote{We let $\bpe^{*}$ denote the reflexive-transitive closure of $\bpe$.} % Using this in the definition implies that $(\pathi,i) \hreach{\homnat} (\pathi,i)$, which will be used to prove the reflexivity of $\leq$ in our counter-model later on.}

\smallskip

\noindent
(2)  $(\pathi,i) \hreach{\sthom} (\pathii,j)$ \iffi there exist $\bptr{\ns}, \bptr{\nsii} \in \bps$ s.t. $(\pathi,i) \in \bptr{\ns}$, $(\pathii,j) \in \bptr{\nsii}$, $\bptr{\ns}$ is a repeat of %with $\bptr{\nsii}$ its companion,
$\bptr{\nsii}$, and $\sthom \in \vechom$ maps $\bptr{\ns}$ to $\bptr{\nsii}$ with $\sthom(\pathi,i) = (\pathii,j)$.

\smallskip

\noindent
We define a \emph{morphic sequence} in $\bp(\nsinp)$ relative to $\vechom$ to be a path
$$
(\pathi_{0}, i_{0}) \hreach{\mor_{1}} (\pathi_{1}, i_{1}) \hreach{\mor_{2}} (\pathi_{2}, i_{2}) \cdots (\pathi_{n-1}, i_{n-1}) \hreach{\mor_{n}} (\pathi_{n}, i_{n})
$$
such that each $\mor_{i}$ is a strong morphism in $\vechom$ or a natural morphism. We write $(\pathi_{0},i_{0}) \hreach{\sigma} (\pathi_{n},i_{n})$ with $\sigma = \mor_{1}, \mor_{2} \ldots, \mor_{n}$ as shorthand. % for such a morphic sequence.

Let $\mdep{\nsinp} = k$. We use $\bp(\nsinp)$ to define the \emph{$\lb \cons, \bp(\nsinp),\vechom \rb$-model} $M = (W, \leq, R,V)$ such that:

\smallskip

\noindent
(1) The set of worlds is defined as follows:
$$
W := \!\!\!\!\!\!\!\!\!\!\!\! \bigcup_{\bptr{\ns} = (\trv,\tre,\trl) \in \bps} \!\!\!\!\!\!\!\!\!\!\!\! \trv
$$

\smallskip

\noindent
(2) $(\pathi,i) \leq (\pathii,j)$ \iffi there exist $\bptr{\ns}, \bptr{\nsii} \in \bps$ s.t. $(\pathi,i) \in \bptr{\ns}$, $(\pathii,j) \in \bptr{\nsii}$, and there exists a morphic sequence $(\pathi,i) \hreach{\sigma} (\pathii,j)$; % in $\bp(\nsinp)$;

\smallskip

\noindent
(3) $R$ is the smallest set satisfying (a) and (b), for all $\bptr{\ns} \in \bps$:

\smallskip

%\begin{description}[leftmargin=!, labelwidth=1.5em, labelsep=.75em, align=left]

%\item[$\quad (a)$] 

\noindent
\quad $(a)$ If $(\pathi,i), (\pathii,i) \in \bptr{\ns}$ and $\bptr{\ns} \models (\pathi,i) \rable^{\cons}_{\forb} (\pathii,i)$, then it follows that $(w,i) R (u,i)$;

%\item[$\quad (b)$] 

\smallskip

\noindent
\quad $(b)$ If $\D \in \cons$, $(\pathi,i) \in \bptr{\ns}$, and $(\pathi,i) \in \depk{\ns}{k+1}$, then it follows that $(\pathi, i) R (\pathi, i)$;

%\end{description}

\smallskip

\noindent
(4) $(\pathi,i) \in V(p)$ \iffi for some $\bptr{\ns} \in \bps$, $p \in \infs{(\pathi,i), \bptr{\ns}}$.

\smallskip

\noindent
This completes the definition of the extracted counter-model.

%\begin{lemma}
%*seems helpful to have a lemma (for counter-model extraction below) stating that saturating a sequent is a finite process in the proof search algorithm
%\end{lemma}

\begin{lemma}\label{lem:bp-monotonicity-witnesses}
Let $\bptr{\ns}\in \bps$ with $(\pathi,i) \in \bptr{\ns}$. The following hold:
\begin{description}[leftmargin=!, labelwidth=1em, labelsep=.5em]

\item[$(1)$] If $\bptr{\nsii} \in \bps$ such that $(\pathii,j) \in \bptr{\nsii}$, $(\pathi,i) \hreach{\sigma} (\pathii,j)$, and $A \in \infs{(\pathi,i),\bptr{\ns}}$, then $A \in \infs{(\pathii,j),\bptr{\nsii}}$;

\item[$(2)$] If $A \iimp B \in \outfs{(\pathi,i),\bptr{\ns}}$, then there is a $\bptr{\nsii} \in \bps$ with $(\pathii,j) \in \bptr{\nsii}$ such that $(\pathi,i) \hreach{\sigma} (\pathii,j)$, $A \in \infs{(\pathii,j),\bptr{\nsii}}$, and $B \in \outfs{(\pathii,j),\bptr{\nsii}}$;

\item[$(3)$] If $\xbox C \in \outfs{(\pathi,i),\bptr{\ns}}$, then there exists a $\bptr{\nsii} \in \bps$ where $(\pathii,j), (\pathiii,j) \in \bptr{\nsii}$, $\bptr{\nsii} \models (\pathii,j) \rable_{\forb}^{\cons} (\pathiii,j)$, $(\pathi,i) \hreach{\sigma} (\pathii,j)$, and $C \in \outfs{(\pathiii,j),\bptr{\nsii}}$;

%\item[$(4)$] If $\bbox D \in \outfs{(\pathi,i),\bptr{\ns}}$, then there exists a $\bptr{\nsii} \in \bps$ with $(\pathii,j), (\pathiii,j) \in \bptr{\nsii}$ such that $\prgr{\bptr{\nsii}} \models (\pathii,j) \rableosb (\pathiii,j)$, $(\pathi,i) \hreach{\sigma} (\pathii,j)$, and $D \in \outfs{(\pathii,j),\bptr{\nsii}}$.

%\item[(3)] If $\wboximpr$ occurs in $\bp(\nsinp)$, $\ns$ is a premise of $\wboximpr$ with $\nsii$ a closest saturated descendant, and $A \in \outfs{w,\ns}$, then $A \in \outfs{w,\nsii}$.

\end{description}
\end{lemma}

\begin{proof} We prove each statement in turn. 

(1) Suppose $\bptr{\nsii} \in \bps$ such that $(\pathii,j) \in \bptr{\nsii}$, $(\pathi,i) \hreach{\sigma} (\pathii,j)$, and $A \in \infs{(\pathi,i),\bptr{\ns}}$. We prove the claim by induction on the number of morphisms in $\sigma$. 

\textit{Base case.} There are two cases: (i) $\sigma = \homnat$ or (ii) $\sigma = \sthom \in \vechom$. We consider case (i) first. By definition, $\pathi = \pathii$ and $\bptr{\ns} \bpe^{*} \bptr{\nsii}$. By the latter fact, $\tr{\ns} \cte^{*} \tr{\nsii}$ with $\cte^{*}$ the reflexive-transitive closure of $\cte$. Therefore, $\tr{\nsii}$ was obtained from $\tr{\ns}$ by (zero or more) bottom-up applications of rules in $\calc$. If one inspects the rules of $\calc$, they will notice that all rules bottom-up preserve input formulae at components. This carries over to $\ct(\nsinp)$ and to $\bp(\nsinp)$ as well. Hence, $A \in \infs{(\pathi,j),\bptr{\nsii}} = \infs{(\pathii,j),\bptr{\nsii}}$.

Let us now consider case (ii). By our assumption, $\bptr{\ns}$ is a repeat with $\bptr{\nsii}$ its companion, and $\sthom \in \vechom$ maps $\bptr{\ns}$ to $\bptr{\nsii}$ with $\sthom(\pathi,i) = (\pathii,j)$. Since strong morphisms preserve input formulae, $A \in \infs{(\pathii,j),\bptr{\nsii}}$.

\textit{Inductive step.} Let $\sigma := \sigma', \mor$. Then, there exists a $\bptr{\nsiii} \in \bps$ with $(\pathiii,k) \in \bptr{\nsiii}$ such that $(\pathi,i) \hreach{\sigma} (\pathii,j)$ is equal to the following:
$$
(\pathi,i) \hreach{\sigma'} (\pathiii,k) \hreach{\mor} (\pathii,j).
$$
By IH, $A \in \infs{(\pathiii,k),\bptr{\nsiii}}$, and by the base case, $A \in \infs{(\pathii,j),\bptr{\nsii}}$.

(2) Suppose $A \iimp B \in \outfs{(\pathi,i),\bptr{\ns}}$. Notice that $\bptr{\ns}$ cannot be $\cons$-stable since it contains $A \iimp B$ as an output formula. This gives two cases: (i) $\bptr{\ns}$ is not a repeat or (ii) it is a repeat.

(i) If $\bptr{\ns}$ is not a repeat, then $\tr{\ns}$ is not a repeat in $\ct(\nsinp)$. Thus, there will be a path from $\tr{\ns}$ in $\ct(\ns)$ such that (a) $\dbr$ is applied bottom-up somewhere along the path, (b) all seq-trees of nested sequents along the path are labeled with $\fval$, and (c) above the $\dbr$ application is a seq-tree $\tr{\nsii}$ such that $\tr{\nsii}$ is $\cons$-saturated. %, $A \in \infs{(w,j),\tr{\nsii}}$, and $B \in \outfs{(w,j),\tr{\nsii}}$. 
Let $\tr{\ns}$ have position $i$ and $\tr{\nsii}$ have position $j$ in $\enum$. From what was said above, $\bptr{\nsii} \in \bps$ with $A \in \infs{(\pathi,j),\tr{\nsii}}$ and $B \in \outfs{(\pathi,j),\tr{\nsii}}$. Moreover, there exists a path in $\bp(\nsinp)$ from $\bptr{\ns}$ to $\bptr{\nsii}$ giving rise to a morphic sequence $(\pathi,i) \hreach{\sigma} (\pathi,j)$ such that $\sigma$ consists solely of natural morphisms. This establishes the claim in this case.

(ii) If $\bptr{\ns}$ is a repeat, then it has an ancestor $\bptr{\nsii}$ and there exists a strong morphism $\sthom$ from $\bptr{\ns}$ to $\bptr{\nsii}$. Hence, $A \iimp B \in \outfs{(\pathii,j),\bptr{\nsii}}$, $\bptr{\nsii}$ is not a repeat, and $(\pathi,i) \hreach{\sthom} (\pathii,j)$. One can argue in a manner similar to case (i) that a $\bptr{\nsiii} \in \bps$ exists with $(\pathii,k) \in \bptr{\nsiii}$ such that $(\pathii,j) \hreach{\sigma} (\pathii,k)$, $A \in \infs{(\pathii,k),\bptr{\nsiii}}$, and $B \in \outfs{(\pathii,k),\bptr{\nsiii}}$. Thus, there exists a morphic sequence $(\pathi,i) \hreach{\sthom,\sigma} (\pathii,k)$ from $\bptr{\ns}$ to $\bptr{\nsiii}$, which establishes the claim.

(3) Suppose $\xbox C \in \outfs{(\pathi,i),\bptr{\ns}}$. Notice that $\bptr{\ns}$ cannot be $\cons$-stable since it contains $\xbox C$ as an output formula. This gives two cases: (i) $\bptr{\ns}$ is not a repeat or (ii) it is a repeat.

(i) If $\bptr{\ns}$ is not a repeat, then $\tr{\ns}$ is not a repeat in $\ct(\nsinp)$. %Based on our assumptions, one can confirm that $\ctl(\tr{\ns}) = \fval$ and $\xbox C \in \outfs{\tail(\pathi),\tr{\ns}}$. 
Thus, there will be a path from $\tr{\ns}$ in $\ct(\ns)$ such that (a) $\dbr$ is applied bottom-up somewhere along the path, (b) all seq-trees of nested sequents along the path are labeled with $\fval$, and (c) above the $\dbr$ application is a seq-tree $\tr{\nsii}$ such that $\tr{\nsii} = (\trv,\tre,\trl)$ is $\cons$-saturated, $w \tre u$ with $\trl(w,u) = \forb$, and $C \in \outfs{u,\tr{\nsii}}$. Let $\tr{\ns}$ have position $i$ and $\tr{\nsii}$ have position $j$ in $\enum$. Observe that the path satisfying conditions (a)--(c) in $\ct(\nsinp)$ gives rise to a morphic sequence $(\pathi,i) \hreach{\sigma} (\pathi,j)$ such that $\sigma$ consists solely of natural morphisms. Moreover, by what was said above $\bptr{\nsii} \models (\pathi,j) \rable_{\forb}^{\cons} (\pathii,j)$. This establishes the claim in this case.

(ii) Suppose $\bptr{\ns}$ is a repeat. Then, it has an ancestor $\bptr{\nsii}$ and there exists a strong morphism $\sthom$ from $\bptr{\ns}$ to $\bptr{\nsii}$. Hence, $\xbox C \in \outfs{(\pathii,j),\bptr{\nsii}}$, $\bptr{\nsii}$ is not a repeat, and $(\pathi,i) \hreach{\sthom} (\pathii,j)$. One can argue in a manner similar to case (i) that a $\bptr{\nsiii} \in \bps$ exists with $(\pathiii,k), (\pathiv,k) \in \bptr{\nsii}$ such that $\prgr{\bptr{\nsii}} \models (\pathiii,k) \rable_{\forb}^{\X} (\pathiv,k)$, $(\pathii,j) \hreach{\sigma} (\pathiii,k)$, and $D \in \outfs{(\pathiv,k),\bptr{\nsii}}$. Thus, there exists a morphic sequence $(\pathi,i) \hreach{\sthom,\sigma} (\pathiii,k)$ from $\bptr{\ns}$ to $\bptr{\nsiii}$, which establishes the claim.
\end{proof}

Before concluding this section and proving that the extracted counter-model is indeed a valid model that falsifies the input $\nsinp$, we first present an illustrative example of a blueprint and the corresponding counter-model obtained from failed proof-search.

\begin{figure}
\centering

\begin{minipage}{.5\columnwidth}
\AxiomC{$ \sar  (\fd) [p \sar q]$}
\RightLabel{$\dbr$}
\UnaryInfC{$ \sar  (\fd) [\ \sar p \imp q]$}
\AxiomC{$r \sar s$}
\RightLabel{$\dbr$}
\BinaryInfC{$ \sar \wbox (p \imp q), r \imp s$}
\RightLabel{$\disr$}
\UnaryInfC{$ \sar \wbox (p \imp q) \lor (r \imp s)$}
\DisplayProof
\end{minipage}
\begin{minipage}{.05\columnwidth} %\textwidth}
\ 
\end{minipage}
\begin{minipage}{.4\columnwidth}
\begin{tikzpicture}
\node[] (a) %[label={[label distance=0cm]270:$\not\sat \xbox (p \imp q) \lor (r \imp s)$}]
[] {$(w,1)$};
\node[] (b) [above=of a,xshift=0mm,yshift=-2.5mm] {$(w,2)$};
\node[] (c) [right=of a,xshift=0mm,yshift=0mm,label={[label distance=0cm]0:$\stackrel{\textstyle \sat r}{\not\sat s}$}] {$(w,4)$};
\node[] (d) [right=of b,xshift=0mm,yshift=0mm,label={[label distance=0cm]0:$\stackrel{\textstyle \not\sat p}{\not\sat q}$}] {$(u,2)$};
\node[] (e) [above=of b,xshift=0mm,yshift=-2.5mm] {$(w,3)$};
\node[] (f) [right=of e,xshift=0mm,yshift=0mm,label={[label distance=0cm]0:$\stackrel{\textstyle \sat p}{\not\sat q}$}] {$(u,3)$};

\path[->] (a) edge[dotted] node[midway, left] {$\leq$} (b);

\path[->] (a) edge[dotted] node[midway, below] {$\leq$} (c);

\path[->] (b) edge[dotted] node[midway, left] {$\leq$} (e);

\path[->] (d) edge[dotted] node[midway, left] {$\leq$} (f);

\path[->] (b) edge[] node[above] {$R$} (d);

\path[->] (e) edge[] node[above] {$R$} (f);

%\draw[->] (b) to [out=0,in=225,looseness=.75] node[midway,below] {}  (d);

\end{tikzpicture}
\end{minipage}
    
\caption{Failed proof-search and extracted counter-model.\label{fig:fail-PS-counter-model-example}}
\end{figure}
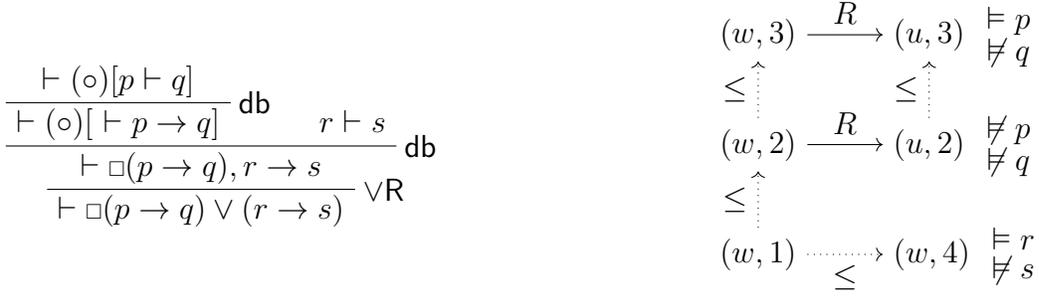

\begin{example}\label{ex:proof-search} Let us consider a run of $\prove(\nsinp)$ with the input $\nsinp = \ \sar \wbox (p \imp q) \lor (r \imp s)$. This generates the computation tree $\ct(\nsinp)$ shown left in Figure~\ref{fig:fail-PS-counter-model-example}. Note that every member of $\ct(\nsinp)$ is labeled with $\fval$. The seq-trees in $\ct(\nsinp)$ are listed below in nested sequent notation:
{\setlength{\multicolsep}{0pt}
\begin{multicols}{2}
\begin{description}[
  leftmargin=!,
  labelwidth=0em,
  labelsep=.25em,
  align=left
]

\item[$\tr{\nsinp}$]
$:=\ \sar \wbox (p \imp q) \lor r \imp s$

\item[$\tr{\ns_{1}}$]
$:=\ \sar \wbox (p \imp q),\ r \imp s$

\item[$\tr{\ns_{2}}$]
$:=\ \sar (\fd)[\ \sar p \imp q]$

\item[$\tr{\ns_{3}}$]
$:=\ \sar (\fd)[p \sar q]$

\item[$\tr{\ns_{4}}$]
$:=\ r \sar s$

\end{description}
\end{multicols}
}

\noindent
Let $\enum = \tr{\ns_{1}}, \tr{\ns_{2}}, \tr{\ns_{3}}, \tr{\ns_{4}}$ be an enumeration of the $\cons$-saturated seq-trees. The blueprint $\bp(\nsinp) = (\bpv,\bps,\bpe)$ of $\prove(\nsinp)$ is defined as follows:

\smallskip

\noindent
(1) $\bpv$ contains all seq-trees in $\ct(\nsinp)$;

\noindent
(2) $\bps := \set{\tr{\ns_{1}}, \tr{\ns_{2}}, \tr{\ns_{3}}, \tr{\ns_{4}}}$;

\noindent
(3) $\bpe := \set{(\tr{\nsinp},\tr{\ns_{1}}), (\tr{\ns_{1}}, \tr{\ns_{2}}),(\tr{\ns_{1}}, \tr{\ns_{4}}),(\tr{\ns_{2}}, \tr{\ns_{3}})}$.

\smallskip

\noindent
Using $\bp(\nsinp)$, we can extract the model $M = (W,\leq,R,V)$ shown right in Figure~\ref{fig:fail-PS-counter-model-example}. Let $\rt$ be the name of the root of $\nsinp$ and $u$ be the name of the non-root components in $\ns_{2}$ and $\ns_{3}$. One can verify that $M$ is indeed a $(\cons,\bp(\ns),\sthom)$-model with $\sthom = \emptyset$ and is based on a frame in $\frames_{\cons}$. Last, $M, (w,1) \not\sat \wbox (p \imp q) \lor r \imp s$.
\end{example}

\begin{theorem}\label{lem:counter-model-is-model}
The $\lb \cons, \bp(\nsinp), \vechom \rb$-model $M = (W, \leq, R,V)$ is a model based on a frame in $\frames_{\cons}$.
\end{theorem}

\begin{proof} Let us prove that $M$ satisfies all properties of a model based on a frame in $\frames_{\cons}$. First, if $\prove(\nsinp) = \false$, then $\prove$ must have generated at least one $\cons$-saturated seq-tree $\tr{\nsinp}$, which means $W \neq \emptyset$ as $\enum\tr{\nsinp}$ must contain at least one vertex. %This latter fact is a consequence of the fact that every tree of a nested sequent has a non-empty set of vertices.

Second, it is trivial to verify that $\leq \ \subseteq W \times W$ and $R \ \subseteq W \times W$. Let us now argue that $\leq$ is reflexive and transitive.

Let $(\pathi,i) \in W$. Then, there exists a $\cons$-saturated seq-tree $\enum\tr{\ns} \in \bps$ with $(\pathi,i) \in \enum\tr{\ns}$. By definition, the following morphic sequence exists in $\bp(\nsinp)$: $(\pathi,i) \hreach{\homnat} (\pathi,i)$. This implies that $(\pathi,i) \leq (\pathi,i)$, which shows that $\leq$ is reflexive. Let us now argue that $\leq$ is transitive. Suppose $(\pathi,i) \leq (\pathii,j)$ and $(\pathii,j) \leq (\pathiii,k)$. By the definition of $\leq$, the two morphic sequences shown below left and below middle exist in $\bp(\nsinp)$.
$$
(\pathi,i) \hreach{\sigma} (\pathii,j)
\qquad
(\pathii,j) \hreach{\sigma'} (\pathiii,k)
\qquad
(\pathi,i) \hreach{\sigma,\sigma'} (\pathiii,k)
$$
As shown above right, we may compose the two given morphic sequences to obtain a new one from $(\pathi,i)$ to $(\pathiii,k)$, implying that $(\pathi,i) \leq (\pathiii,k)$. Hence, $\leq$ is transitive. We now argue that $M$ satisfies all frame conditions in $\cons$, (F1), and (F2). It is simple to show that $M$ satisfies all conditions in $\cons$, so we show that it satisfies condition (F1) and note that the argument for (F2) is similar.

(F1). Suppose $(\pathi,i) \leq (\pathii,j)$ and $(\pathi,i) R (\pathiii,i)$. Then, there exist $\bptr{\ns}, \bptr{\nsii} \in \bps$ with $(\pathi,i), (\pathiii,i) \in \bptr{\ns}$ and $(\pathii,j) \in \bptr{\nsii}$. By our first supposition, we know there exists a morphic sequence of the form $(\pathi,i) \hreach{\sigma} (\pathii,j)$ in $\bp(\nsinp)$. Every morphism in $\sigma$ is at least a weak morphism, and so, by \lem~\ref{lem:morpic-properties}-(3) we know that $\mor_{\sigma}$ is a weak morphism from $\bptr{\ns}$ to $\bptr{\nsii}$ such that $\mor_{\sigma}(\pathi,i) = (\pathii,j)$. Since $(\pathi,i) R (\pathiii,i)$, we have that $\bptr{\ns} \models (\pathi,i) \rable_{\forb}^{\cons} (\pathii,i)$. By \lem~\ref{lem:morpic-properties}-(4), it follows that $\bptr{\nsii} \models \mor_{\sigma}(\pathi,i) \rable_{\forb}^{\cons} \mor_{\sigma}(\pathii,i)$. Let $\mor_{\sigma}(\pathii,i) = (\pathiii',j)$. By the definition of $R$ then, we have $(\pathii,j) R (\pathiii',j)$. Also, $(\pathiii,i) \leq (\pathiii',j)$ because $(\pathiii,i) \hreach{\sigma} (\pathiii',j)$, which establishes (F1).

%(F2) This case is similar to case (F1) above.

Last, we show that the model $M$ satisfies the monotonicity property (M). Suppose that $(\pathi,i) \in V(p)$ and $(\pathi,i) \leq (\pathii,j)$. By the definition of $V$, there exists a $\bptr{\ns} \in \bps$ with $p \in \infs{(\pathi,i), \bptr{\ns}}$. By definition of $\leq$ there exists a $\bptr{\nsii} \in \bps$ with $(\pathii,j) \in \bptr{\nsii}$ such that the morphic sequence $(\pathi,i) \hreach{\sigma} (\pathii,j)$ exists in $\bp(\nsinp)$. By Lemma~\ref{lem:bp-monotonicity-witnesses}-(1), $p \in \infs{(\pathii,j), \bptr{\nsii}}$, and so, $(\pathii,j) \in V(p)$.
\end{proof}

%\begin{lemma}\label{lem:cluster-hom-eq}
%(1) All nested seq in a cluster are homormorphically equivalent.
%(2) If two nested seq are hom equivalent, then prop path exists in one iffi it exists in the other
%\end{lemma}

\begin{theorem}\label{thm:false-implies-not-prove}
If $\prove(\nsinp) = \false$, then one can construct a model $M = (W, \leq, R ,V)$ such that $M \not\sat \nsinp$, i.e., $\calc \not\proves \nsinp$.
\end{theorem}

\begin{proof} Let $\prove(\nsinp) = \false$ and $\ct(\nsinp) = (\ctv,\cte,\ctl)$ be the corresponding computation tree. Using this computation tree, one can define the blueprint $\bp(\nsinp) = (\bpv,\bps,\bpe)$, from which the $\lb \cons, \bp(\nsinp), \sthom \rb$-model $M = (W, \leq, R ,V)$ can be constructed. We know that $M$ is a model based on a frame in $\frames_{\cons}$ by Lemma~\ref{lem:counter-model-is-model}.

We show by simultaneous induction that if $\bptr{\ns} \in \bps$ with $(\pathi,i) \in \bptr{\ns}$, then (1) if $A \in \infs{(\pathi,i),\bptr{\ns}}$, then $M, (\pathi,i) \sat A$ and (2) if $B \in \outfs{(\pathi,i),\bptr{\ns}}$, then $M, (\pathi,i) \not\sat B$. We prove the claim for the most interesting cases and remark that all remaining cases are simple or similar.

\noindent
$\bullet$ $p \in \infs{(\pathi,i),\bptr{\ns}}$. In this case, $(\pathi,i) \in V(p)$ by definition, showing $M, (\pathi,i) \sat p$.

\noindent
$\bullet$ $p \in \outfs{(\pathi,i),\bptr{\ns}}$. Suppose $p \in \infs{(\pathi,i),\bptr{\ns}}$ for a contradiction. Then, $p \in \infs{\bptr{\ns}} \cap \outfs{\bptr{\ns}}$, meaning, condition $\idc$ is not satisfied (see \dfn~\ref{def:saturated-seq}). Hence, $\bptr{\ns}$ is not $\cons$-saturated, contrary to our assumption. Therefore, $p \not\in \infs{(\pathi,i),\bptr{\ns}}$, meaning $(\pathi,i) \not\in V(p)$. We conclude that $M, (\pathi,i) \not\sat p$.

\noindent
$\bullet$ $A \imp B \in \infs{(\pathi,i),\bptr{\ns}}$. Let $(\pathi,i) \leq (\pathii,j)$ and $M, (\pathii,j) \sat A$. By the definition of $\leq$, there exists a $\bptr{\nsii} \in \bps$ such that $(\pathii,j) \in \bptr{\nsii}$, and there exists a morphic sequence $(\pathi,i) \hreach{\sigma} (\pathii,j)$ in $\bp(\nsinp)$. By Lemma~\ref{lem:bp-monotonicity-witnesses}, we know that $A \imp B \in \infs{(\pathii,j),\bptr{\nsii}}$. Suppose for a contradiction $B \not\in \infs{(\pathii,j),\bptr{\nsii}}$. Since $\bptr{\nsii}$ is $\cons$-saturated it must be the case that $A \in \outfs{(\pathii,j),\bptr{\nsii}}$ by condition $\iimplc$. By IH, we have that $M, (\pathii,j) \not\sat A$, yielding a contradiction with our assumption. Hence, $B \in \infs{(\pathii,j),\bptr{\nsii}}$, which implies that $M, (\pathii,j) \sat B$ by IH, and so, $M, (\pathi,i) \sat A \iimp B$ since $(\pathii,j)$ was arbitrary.

\noindent
$\bullet$ $A \imp B \in \outfs{(\pathi,i),\bptr{\ns}}$. By Lemma~\ref{lem:bp-monotonicity-witnesses}, we know there exists a $\bptr{\nsii} \in \bps$ with $(\pathii,j) \in \bptr{\nsii}$ such that $A \in \infs{(\pathii,j),\bptr{\nsii}}$, $B \in \outfs{(\pathii,j),\bptr{\nsii}}$, and the morphic sequence $(\pathi,i) \hreach{\sigma} (\pathii,j)$ exists in $\bp(\nsinp)$. By the last fact, we know that $(\pathi,i) \leq (\pathii,j)$, and by IH, we know that $M, (\pathii,j) \sat A$ and  $M, (\pathii,j) \not\sat B$. Consequently,  $M, (\pathi,i) \not\sat A \iimp B$.

\noindent
$\bullet$ $\wdia A \in \infs{(\pathi,i),\bptr{\ns}}$. By the $\xdialc$ condition, we know there exists a $(\pathii,i) \in \bptr{\ns}$ such that $(\pathi,i) \tre' (\pathii,i)$, $\trl'((\pathi,i),(\pathii,i)) = \fd$, and $A \in \infs{(\pathii,i),\bptr{\ns}}$ with $\bptr{\ns} = (\trv',\tre',\trl')$. By the definition of $R$ and IH, it follows $(\pathi,i) R (\pathii,i)$ and $M, (\pathii,i) \sat A$, which implies $M, (\pathi,i) \sat \wdia A$.

\noindent
$\bullet$ $\wdia A \in \outfs{(\pathi,i),\bptr{\ns}}$. Suppose $(\pathi,i) R (\pathii,i)$. Then, by the definition of $R$, $\bptr{\ns} \models (\pathi,i) \rable_{\fd}^{\cons} (\pathii,i)$. Since $\bptr{\ns}$ is $\cons$-saturated, we know that $A \in \outfs{(\pathii,i),\bptr{\ns}}$ by condition $\xdiarc$. By IH, $M, (\pathii,i) \not\sat A$, so $M, (\pathi,i) \not\sat \wdia A$ since $(\pathii,i)$ was arbitrary.

\noindent
$\bullet$ $\wbox A \in \outfs{w_{i},\nsii}$. By Lemma~\ref{lem:bp-monotonicity-witnesses}, we know there exists a $\bptr{\nsii} \in \bps$ with $(\pathii,j), (\pathiii,j) \in \bptr{\nsii}$ such that $A \in \outfs{(\pathiii,j),\bptr{\nsii}}$ and the following two facts hold:
$$
(\pathi,i) \hreach{\sigma} (\pathii,j)
\qquad
\bptr{\nsii} \models (\pathii,j) \rableosf^{\cons} (\pathiii,j)
$$
Hence, $(\pathi,i) \leq (\pathii,j)$ and $(\pathii,j) R (\pathiii,j)$. By IH, $M, (\pathii,j) \not\sat A$, showing that $M, (\pathi,i) \not\sat \wbox A$.
\end{proof}

As a consequence of Theorems~\ref{thm:success-gives-proof}, \ref{thm:false-implies-not-prove}, and \ref{thm:termination}, we obtain: %we immediately obtain the following.

\begin{theorem}
For $\cons \subseteq \set{\T,\B,\D}$, $\iktc$ has the finite model property and is decidable.
\end{theorem}

%%%CONCLUSION
\section{Concluding Remarks}\label{sec:conclusion}

%Notes:
%1. Modal dynamic blocking method used to presernve numebrs restrictions, so this method would work for int modal logics with graded moalities most likely
%2. Can generalize to int gram logics easily and discover far more expressive classes with int path axioms
%3. Work on how to fold models into finite counter models (i.e. how to get FMP)

There are several directions for future work. First, our proof-search methodology may be generalizable to certain classes of intuitionistic grammar logics (IGLs). IGLs were introduced as multi-modal generalizations of both intuitionistic modal and intuitionistic tense logics~\cite{Lyo21b}, and are the intuitionistic counterparts of classical grammar logics~\cite{DemNiv05,TiuIanGor12}. It seems likely that our algorithm can be extended to accommodate specific extensions of %the intuitionistic multi-modal logic 
$\mathsf{IK_m}$.

Second, it would be worthwhile to investigate whether the complexity of our procedure can be reduced, as our proof-search algorithm has super-exponential worst-case complexity due to the necessity of loop-checking. One possible direction is to employ linear nested sequents during proof-search, as has been done for classical tense logics~\cite{GorLel19}. This may yield smaller computation trees and, consequently, improved complexity bounds.

Third, the decidability of intuitionistic modal logics with transitive modal relations (e.g., $\mathsf{IK4}$ and $\mathsf{IS4}$) has remained a longstanding open problem. The techniques developed in this paper address core obstacles that arise in proof-search for such logics, particularly those stemming from loop-checking and non-invertibility. Due to the modularity of our nested sequent systems, it is straightforward to transform them into systems for transitive IMLs (cf.~\cite{Lyo25}), yielding a natural setting in which to study proof-search and decidability.

%There are a few immediate avenues for future work. First, one could attempt to generalize our proof-search method to cover certain \emph{intuitionistic grammar logics (IGLs)}. IGLs were introduced as a multi-modal generalization of both IMLs and ITLs~\cite{Lyo21b}, and are the intuitionistic counter-part of classical grammar logics~\cite{DemNiv05,TiuIanGor12}. It seems that one can generalize our proof-search algorithm to accommodate specific extensions of $\mathsf{IK_{m}}$. Second, the decidability of IMLs with a transitive modal relation (e.g., $\mathsf{IK4}$ and $\mathsf{IS4}$) has been a longstanding open problem. The work in this paper addresses core issues that arise when attempting to establish the decidability of such logics via proof-search. Due to the modularity of our nested systems, it is straightforward to transform them into systems for transitive IMLs (cf.~\cite{Lyo25}) and study proof-search in this setting. Third, it would be worthwhile to investigate how to reduce the complexity (if possible) of our proof-search method, e.g., by employing linear nested sequents during proof-search, as has been done in the setting of classical tense logics~\cite{GorLel19}. This could reduce the size of computation trees, and thus, reduce the complexity of proof-search.

% REFERENCES

\bibliography{bibliography}

\appendix

\section{Supplement for Section~\ref{sec:systems}}

\paragraph{Modal Rules.} Below, we have included explicit presentations of the modal rules $\xdial$, $\xdiar$, $\xboxl$, and $\xboxr$ depending on the value of the parameters $\forb \in \set{\fd,\bd}$. Note that $\wboxl$, $\bboxl$, $\wdiar$, and $\bdiar$ are subject to the side condition $\dag(\cons)$.
\begin{center}
\begin{tabular}{c @{\hskip .25em} c}
\AxiomC{$\ns \sbl \Gamma \sar \Delta, (\fd) [ A \sar \emptyset ] \sbr$}
\RightLabel{$\wdial$}
\UnaryInfC{$\ns \sbl \Gamma, \wdia A \sar \Delta \sbr$}
\DisplayProof

&

\AxiomC{$\ns \sbl \Gamma \sar \wdia A, \Delta \sbr_{w} \sbl \Sigma \sar A, \Pi \sbr_{u}$}
\RightLabel{$\wdiar$}
\UnaryInfC{$\ns \sbl \Gamma \sar \wdia A, \Delta \sbr_{w} \sbl \Sigma \sar \Pi \sbr_{u}$}
\DisplayProof
\end{tabular}
\end{center}

\smallskip

\begin{center}
\begin{tabular}{c @{\hskip .25em} c}
\AxiomC{$\ns \sbl \Gamma, \wbox A \sar \Delta \sbr_{w} \sbl \Sigma, A \sar \Pi \sbr_{u}$}
\RightLabel{$\wboxl$}
\UnaryInfC{$\ns \sbl \Gamma, \wbox A \sar \Delta \sbr_{w} \sbl \Sigma \sar \Pi \sbr_{u}$}
\DisplayProof

&

\AxiomC{$\ns^{\downarrow} \sbl \Gamma \sar (\fd) [ \emptyset \sar A ] \sbr$}
\RightLabel{$\wboxr$}
\UnaryInfC{$\ns \sbl \Gamma \sar \wbox A, \Delta \sbr$}
\DisplayProof
\end{tabular}
\end{center}

\smallskip

\begin{center}
\begin{tabular}{c @{\hskip .25em} c}
\AxiomC{$\ns \sbl \Gamma \sar \Delta, (\bd) [ A \sar \emptyset ] \sbr$}
\RightLabel{$\bdial$}
\UnaryInfC{$\ns \sbl \Gamma, \bdia A \sar \Delta \sbr$}
\DisplayProof

&

\AxiomC{$\ns \sbl \Gamma \sar \bdia A, \Delta \sbr_{u} \sbl \Sigma \sar A, \Pi \sbr_{w}$}
\RightLabel{$\bdiar$}
\UnaryInfC{$\ns \sbl \Gamma \sar \bdia A, \Delta \sbr_{u} \sbl \Sigma \sar \Pi \sbr_{w}$}
\DisplayProof
\end{tabular}
\end{center}

\smallskip

\begin{center}
\begin{tabular}{c @{\hskip .25em} c}
\AxiomC{$\ns \sbl \Gamma, \bbox A \sar \Delta \sbr_{u} \sbl \Sigma, A \sar \Pi \sbr_{w}$}
\RightLabel{$\bboxl$}
\UnaryInfC{$\ns \sbl \Gamma, \bbox A \sar \Delta \sbr_{u} \sbl \Sigma \sar \Pi \sbr_{w}$}
\DisplayProof

&

\AxiomC{$\ns^{\downarrow} \sbl \Gamma \sar (\bd) [ \emptyset \sar A ] \sbr$}
\RightLabel{$\bboxr$}
\UnaryInfC{$\ns \sbl \Gamma \sar \bbox A, \Delta \sbr$}
\DisplayProof
\end{tabular}
\end{center}

\begin{customthm}{\ref{thm:soundness-nested}} If $\calc \proves \ns$, then $\ns$ is $\fcons$-valid.
\end{customthm}

\begin{proof} It is straightforward to show that all instances of $\id$ and $\botl$ are valid. Therefore, let us consider the other rules of $\calc$. To show soundness, we assume the conclusion is $\frames_{\cons}$-invalid, and argue that at least one premise is $\frames_{\cons}$-invalid. We present the $\iimpr$ and $\xdiar$ cases; the remaining cases are similar. Let $M = (W, \leq,R,V)$ be a model based on a frame in $\frames_{\cons}$ with $\mint$ an $M$-interpretation.

%, (\charx_{1})[\nsiii_{1}], \ldots, (\charx_{n})[\nsiii_{n}]

$\iimpr$. Let $\nsii := \ns \sbl \Gamma \sar A \iimp B, \Delta \sbr_{w}$ be the conclusion of an $\iimpr$ instance and let $\tr{\nsii} := (\trv, \tre, \trl)$. Suppose $M, \mint \not\sat \nsii$. Then, $M, \mint(w) \not\sat A \iimp B$, which implies that there exists a world $w' \in W$ such that $\mint(w) \leq w'$, $M, w' \sat A$, and $M, w' \not\sat B$. We now define a new interpretation $\mint'$ such that $M, \mint' \not\sat \ns^{\downarrow} \sbl \Gamma, A \sar B \sbr$. 

%First, we set $\mint'(w) = u$. Second, we let $v_{1}, \ldots, v_{n}$ be (the names of) all children of $w$, that is, $v_{i}$ is the name of the root of $\nsiii_{i}$ for $i \in [n]$. 

First, we set $\mint'(w) = w'$. Second, we let $v_{1}, \ldots, v_{n}$ be (the names of) all children of $w$ (if they exist) with $u$ the (name of the) parent of $w$ (if it exists) in $\nsii$. Let $\trl(w,v_{i}) = \forb_{i}$ for $i \in [n]$ and $\trl(u,w) = \forb$. Since $\mint(w) R_{\forb_{i}} \mint(v_{i})$ and $\mint(u) R_{\forb} \mint(w)$ hold in $M$ and $\mint(w) \leq w'$, we know there exists a $v_{i}'$ and $u'$ such that $\mint(v_{i}) \leq v_{i}'$ and $w'R_{\forb_{i}}v_{i}'$ by conditions (F1) and (F2), and $\mint(u) \leq u'$ and $u'R_{\forb}w'$ by conditions (F1) and (F2). Thus, if we set $\mint'(v_{i}) = v_{i}'$ and $\mint'(u) = u'$, then it follows that $\mint'(w) R_{\forb_{i}} \mint'(v_{i})$ and $\mint'(u') R_{\forb} \mint'(w)$. Moreover, since $M, \mint \sat \bigwedge \Gamma_{i}$ and $M, \mint \sat \bigwedge \Sigma$ for $\trl(v_{i}) = \Gamma_{i} \sar \Delta_{i}$ and $\trl(u) = \Sigma \sar \Pi$, it follows that $M, \mint' \sat \bigwedge \Gamma_{i}$ and $M, \mint' \sat \bigwedge \Sigma$ by \lem~\ref{lem:persistence} as $\mint(v_{i}) \leq \mint'(v_{i})$ and $\mint(u) \leq \mint'(u)$. We successively repeat this process until all components of $\nsii$ have been processed. One can confirm that $M, \mint' \not\sat \ns^{\downarrow} \sbl \Gamma, A \sar B \sbr$.

$\xdiar$. Let $\nsii = \ns \sbl \Gamma \sar \xdia A, \Delta \sbr_{w} \sbl \Sigma \sar \Pi \sbr_{u}$ and $\tr{\nsii}^{\cons} = (\trv_{\cons},\tre_{\cons},\trl_{\cons})$. Suppose $M, \mint \not\sat \nsii$. Then, $M, \mint(w) \not\sat \xdia A$. By the side condition on $\xdiar$, we know that $\tr{\nsii} \models w \rable_{\forb}^{\cons} u$. We have four cases to consider depending on if $\T$ or $\B$ are members of $\cons$ or not. We consider two cases and note that the remaining cases are similar. 

(1) If $\B,\T \not\in \cons$, then $w \tre u$ with $\trl(w,u) = \forb$, so by Definition~\ref{def:sequent-semantics}, we know that $wRu$. Hence, $M, \mint(u) \not\sat A$, showing that $M, \mint \not\sat \ns \sbl \Gamma \sar \xdia A, \Delta \sbr_{w} \sbl \Sigma, A \sar \Pi \sbr_{u}$.

(2) If $\T \in \cons$ and $\B \not\in \cons$, then either (i) $w \tre u$ with $\trl(w,u) = \forb$ or (ii) $w = u$, $w \tre_{\cons} w$, and $\trl(w,w) = \forb$. Case (i) is resolved as above. In case (ii), observe that since $R$ is reflexive, we have $wRw$, so $M, \mint(w) \not\sat A$. Therefore, $M, \mint \not\sat \ns \sbl \Gamma \sar \xdia A, \Delta \sbr_{w} \sbl \Sigma, A \sar \Pi \sbr_{u}$.
\end{proof}

%%%%%%%%%%%%%%%
%%%%%%%%%%%%%%%
%%%%%%%%%%%%%%%

\section{Supplement for Section~\ref{subsec:comp-trees}}

\paragraph{Computation Tree.} For the sake of completeness, we add the formal definition of the computation tree $\ct(\nsinp) = (\ctv,\cte,\ctl)$ corresponding to $\prove(\nsinp)$. We define $(\ctv,\cte)$ root-first based on the number of recursive calls in $\prove(\nsinp)$. Initially, our structure is taken to be $(\ctv,\cte) := (\set{\nsinp},\emptyset)$. Once the finite structure $(\ctv,\cte)$ has been built, we define (1) $\ctl(\tr{\ns}) = \tval$ \iffi $\prove(\ns) = \true$ and (2) $\ctl(\tr{\ns}) = \fval$ \iffi $\prove(\ns) = \false$,  for all $\tr{\ns} \in \ctv$.
\begin{description}[leftmargin=!, labelwidth=0em, labelsep=.25em]

\item[$\bullet$] If lines 1-3 are executed, then stop building $(\ctv,\cte)$ along the branch ending at the initial sequent;

\item[$\bullet$] If lines 4-6 are executed, then stop building $(\ctv,\cte)$ along the branch ending at the repeat or stable sequent;

\item[$\bullet$] If lines 7-10, 17-20, or 21-24 are executed, then set $\ctv := \ctv \cup \set{\tr{\ns_{1}}, \tr{\ns_{2}}}$ and $\cte \ := \ \cte \cup \set{(\tr{\ns},\tr{\ns_{1}}),(\tr{\ns},\tr{\ns_{2}})}$;

\item[$\bullet$] If lines 11-13, 14-16, 25-27, 28-30, 31-33, or 34-37 are executed, then set $\ctv := \ctv \cup \set{\ns}$ and $\cte \ := \ \cte \cup \set{(\tr{\ns},\tr{\ns'})}$;

\item[$\bullet$] If lines 38-44 are executed, then we set the vertices to be $\ctv := \ctv \cup \set{\tr{\ns_{1}}, \ldots, \tr{\ns_{n{+}k}}}$ and we set the edges to be $\cte \ := \ \cte \cup \set{(\tr{\ns},\tr{\ns_{1}}), \ldots, (\tr{\ns},\tr{\ns_{n{+}k}})}$;

\end{description}

%\section{Supplement for Section~\ref{subsec:repeats-termination}}

%\section{Supplement for Section~\ref{subsec:counter-model-extract}}

\end{document}